\documentclass[letterpaper, 10 pt, journal, romanappendices, twoside]{IEEEtran}

\usepackage{amsmath, amssymb}
\usepackage{graphicx}
\usepackage[sans]{dsfont}
\usepackage{stmaryrd}
\usepackage{mathrsfs}
\usepackage{subfigure}
\usepackage{color}
\usepackage{cite}
\usepackage{asymptote}
\usepackage{enumerate}

\usepackage{calligra}
\DeclareFontShape{T1}{calligra}{m}{n}{<->s*[2.2]callig15}{}
\DeclareMathAlphabet{\mathcalligra}{T1}{calligra}{m}{n}
\DeclareMathAlphabet{\mpzc}{OT1}{pzc}{m}{it}

\newcommand{\mrm}{\mathrm}

\newcommand{\msc}{\mathscr}

\newcommand{\emp}{\emph}

\newcommand{\tends}{\rightarrow}

\def\vect#1{\underline{#1}}

\newtheorem{thm}{Theorem}
\newtheorem{lem}[thm]{Lemma}
\newtheorem{prop}[thm]{Proposition}

\newtheorem{defn}[thm]{Definition}

\newtheorem{rem}{Remark}
\newenvironment{dis}[1][Discussion :]{\begin{trivlist}
\item[\hskip \labelsep {\itshape #1} \hspace*{1mm}]}{\end{trivlist}}

\title{Windowed Decoding of Spatially Coupled Codes}
\author{Aravind~R.~Iyengar, Paul~H.~Siegel,~\IEEEmembership{Fellow,~IEEE},\\R\"{u}diger~L.~Urbanke and Jack~K.~Wolf,~\IEEEmembership{Life~Fellow,~IEEE}
\thanks{A. R. Iyengar was with the Department of Electrical and Computer Engineering and the Center for Magnetic Recording Research, University of California, San Diego. He is now with Qualcomm Technologies Inc., Santa Clara, CA 95051 (email: ariyengar@qti.qualcomm.com). P. H. Siegel is with the Department of Electrical and Computer Engineering and the Center for Magnetic Recording Research, University of California, San Diego, La Jolla, CA 92093 USA (e-mail: psiegel@ucsd.edu). J. K. Wolf (deceased) was with the Department of Electrical and Computer Engineering and the Center for Magnetic Recording Research, University of California, San Diego, La Jolla, CA 92093 USA.}
\thanks{R. L. Urbanke is with the School for Computer and Communication Sciences, Ecole Polytechnique Federale de Lausanne, CH-1015, Switzerland (e-mail: ruediger.urbanke@epfl.ch).}
\thanks{This work was supported in part by the Center for Magnetic Recording Research, by the National Science Foundation under the Grant CCF-$0829865$ and by grant number $200021-121903$ of the Swiss National Foundation.}
\thanks{A summary of the results of this paper were presented at the IEEE International Symposium on Information Theory 2011, St. Petersburg, Russia~\cite{iye_11_isit_wd}.}
}

\newcommand{\twobibs}[2]{#2}
\newcommand{\figselect}[2]{#2}

\begin{document}
\maketitle
\thispagestyle{plain}

\begin{abstract}
Spatially coupled codes have been of interest recently owing to their superior performance over memoryless binary-input channels. The performance is good both asymptotically, since the belief propagation thresholds approach the Shannon limit, as well as for finite lengths, since degree-$2$ variable nodes that result in high error floors can be completely avoided. However, to realize the promised good performance, one needs large blocklengths. This in turn implies a large latency and decoding complexity. For the memoryless binary erasure channel, we consider the decoding of spatially coupled codes through a windowed decoder that aims to retain many of the attractive features of belief propagation, while trying to reduce complexity further. We characterize the performance of this scheme by defining thresholds on channel erasure rates that guarantee a target erasure rate. We give analytical lower bounds on these thresholds and show that the performance approaches that of belief propagation exponentially fast in the window size. We give numerical results including the thresholds computed using density evolution and the erasure rate curves for finite-length spatially coupled codes.
\end{abstract}

\begin{keywords}
Low-density parity-check codes, Belief propagation, Erasure channels, Spatial coupling, Windowed decoding, Iterative decoding.
\end{keywords}

\section{Introduction}
\PARstart S{parse} graph codes have been of great interest in the coding community for close to two decades, after it was shown that statistical inference techniques on graphical models representing these codes had decoding performance that surpassed that of the best known codes. One class of such codes are low-density parity-check (LDPC) codes, which although introduced by Gallager in the 60's \cite{gal_63_the_ldpc} were rediscovered in the 90's after the advent of Turbo Codes \cite{ber_93_icc_turbo} and iterative decoding. Luby et al. showed \cite{lub_01_tit_erco, lub_01_tit_irrldpc} that a decoder based on \emph{belief propagation} (BP) \cite{pea_88_bok_reason} had very good performance for these codes over the binary erasure channel (BEC). This superior performance of LDPC codes was shown by Richardson and Urbanke \cite{ric_01_tit_capldpc} to be true over a broader class of binary-input, memoryless, symmetric-output (BMS) channels. Furthermore these codes were optimized to approach capacity on many of these BMS channels \cite{amr_06_the_ldpcopt, amr_09_url_ldpcopt}.

The convolutional counterparts of LDPC block codes were first introduced by Felstrom and Zigangirov in \cite{fel_99_tit_ldpccc}. There is considerable literature on the constructions and analysis of these ensembles \cite{eng_99_ppi_ldpccc, eng_99_lnc_ldpccc, len_01_ppi_ldpccc, tan_04_tit_ldpccirc}. The BP thresholds for these ensembles were reported in \cite{sri_04_all_ldpcccconv} and shown to be close to capacity in \cite{len_10_tit_ldpcccthresh}. In \cite{len_09_ita_asympreg} the authors construct regular LDPC convolutional codes based on \emph{protographs} \cite{tho_03_jpl_proto} that have BP thresholds close to capacity. In \cite{kud_11_tit_ldpccc}, Kudekar et al. considered convolutional-like codes which they called \emph{spatially coupled codes} and showed that the BP thresholds of these codes approached the MAP thresholds of the underlying unstructured ensembles over the BEC. This observation was made for protograph-based generalized LDPC codes in \cite{len_10_isit_pggldpc}. Evidence for similar results over general BMS channels was given in \cite{kud_10_ist_bp4map}, and proven recently in \cite{kud_12_arx_scbms}. Moreover this phenomenon, termed \emph{threshold saturation}, was shown to be a more generic effect of coupling by showing an improvement in performance of systems based on other graphical models: the random $K$-SAT, $Q$-COL problems from computation theory, Curie-Weiss model from statistical mechanics \cite{has_10_itw_coupled}, and LDGM and rateless code ensembles \cite{are_11_itw_ldgmcoup}. Non-binary LDPC codes obtained through coupling have also recently been investigated \cite{uch_10_arx_nbldpccc}.

The good performance of spatially coupled codes is apparent when both the blocklength of individual codes and the coupling length becomes large. However, as either of these parameters becomes large, BP decoding becomes complex. We therefore consider a \emp{windowed decoder} that exploits the structure of the coupled codes to reduce the decoding complexity while maintaining the advantages of the BP decoder in terms of performance. An additional advantage of the windowed decoder is the reduced latency of decoding. The windowed decoding scheme studied here is the one used to evaluate the performance of protograph-based codes over erasure channels with and without memory \cite{pap_10_itw_wed, iye_10_icc_ldpccc, iye_12_tit_ldpccc}. The main result of this paper is that the windowed decoding thresholds approach the BP thresholds exponentially in the size of the window $W$. Since the BP thresholds are themselves close to the MAP thresholds for spatially coupled codes, windowed decoding thus gives us a way to achieve close to ML performance with complexity reduced further beyond that of the BP decoder.

The rest of the paper is organized as follows. Section \ref{sec_scc} gives a brief introduction to spatially coupled codes. In Section \ref{sec_wd} we discuss the windowed decoding scheme. We state here the main result of the paper which we prove in Section \ref{sec_an}. We give some numerical results in Section \ref{sec_expres} and conclude in Section \ref{sec_conc}. Much of the terminology and notation used in the paper is reminiscent of the definitions in \cite{kud_11_tit_ldpccc} and we often refer the reader to this paper.

\section{Spatially Coupled Codes} \label{sec_scc}
We describe the $(d_l, d_r, \gamma, L)$ spatially coupled ensemble that was introduced in \cite{kud_11_tit_ldpccc} in terms of its Tanner graph. There are $M$ variable nodes at each position in $[L] \triangleq\{1, 2, \cdots, L\}$. We will assume that there are $M\frac{d_l}{d_r}$ check nodes at every integer position, but only some of these interact with the variable nodes. The variable (check) nodes at position $i$ constitute the $i^\text{th}$ \emp{section} of variable (check, resp.) nodes in the code. The $L$ sections of variables are together referred to as the \emph{chain} and $L$ is called the \emph{chain length}. For each of the $d_l$ edges incident on a variable at position $i$, we first choose a section uniformly at random from the set $\{i, i + 1, \cdots, i + \gamma - 1\}$, then choose a check uniformly at random from the $M\frac{d_l}{d_r}$ checks in the chosen section, and connect the variable to this check. We refer to the parameter $\gamma$ as the \emp{coupling length}. It can be shown that this procedure amounts roughly to choosing each of the $d_r$ connections of a check node at position $i$ uniformly and independently from the set $\{i - \gamma + 1, i - \gamma + 2, \cdots, i\}$. Observe that when $\gamma = 1$ this procedure gives us $L$ copies of the $(d_l, d_r)$-regular uncoupled ensemble. Since we are interested in coupled ensembles, we will henceforth assume that $\gamma > 1$. Further, we will typically be concerned with this ensemble when $L \gg \gamma$, in which case the \emp{design rate} given by \cite{kud_11_tit_ldpccc}
\[
R(d_l, d_r, \gamma, L) = 1 - \frac{d_l}{d_r}\Big(1 + O(\frac{\gamma}{L})\Big) 
\]
is close to $1 - \frac{d_l}{d_r}$.

\subsection*{BP Performance} \label{ssec_scc_bp}
In the following we will briefly state known results that are relevant to this work. See \cite{kud_11_tit_ldpccc} for detailed analysis of the BP performance of spatially coupled codes. The BP performance of the $(d_l, d_r, \gamma, L)$ spatially coupled ensemble when $M \tends \infty$ can be evaluated using \emp{density evolution}. Denote the average erasure probability of a message from a variable node at position $i$ as $x_i$. We refer to the vector $\vect{x} = (x_1, x_2, \cdots, x_L)$ as the \emp{constellation}.

\begin{defn}[BP Forward Density Evolution]\vspace*{3mm}
Consider the BP decoding of a $(d_l, d_r, \gamma, L)$ spatially coupled code over a BEC with channel erasure rate $\epsilon$. We can write the forward density evolution (DE) equation as follows. Set the initial constellation to be $\vect{x}^{(0)} = (1, 1, \cdots, 1)$ and evaluate the constellations $\{\vect{x}^{(\ell)}\}_{\ell = 1}^{\infty}$ according to
\begin{equation} \label{eq_fdebp}
x_i^{(\ell)} = 
\begin{cases}
0, \text{if }i \notin [L]{\ }\forall{\ }\ell\text{, and otherwise }\\
\epsilon\Big( 1 - \frac{1}{\gamma}\sum_{j = 0}^{\gamma - 1}(1 - \frac{1}{\gamma}\sum_{k = 0}^{\gamma - 1}x_{i + j - k}^{(\ell - 1)})^{d_r - 1}\Big)^{d_l - 1}.
\end{cases}
\end{equation}
This is called the \emph{parallel schedule} of the BP forward density evolution.\hfill$\square$
\end{defn}

\vspace*{3mm}For ease of notation, we will write the equation
\[
x_i = \epsilon\Big( 1 - \frac{1}{\gamma}\sum_{j = 0}^{\gamma - 1}(1 - \frac{1}{\gamma}\sum_{k = 0}^{\gamma - 1}x_{i + j - k})^{d_r - 1}\Big)^{d_l - 1}
\]
as
\begin{equation} \label{eq_shortde}
x_i = \epsilon g(x_{i - \gamma + 1}, \cdots, x_{i + \gamma - 1}).
\end{equation}
It is clear that the function $g(\cdot)$ is monotonic in each of its arguments.

\begin{defn}[FP of BP Forward DE]\vspace*{3mm}
Consider the parallel schedule of the BP forward DE for the $(d_l, d_r, \gamma, L)$ spatially coupled code over a BEC with erasure rate $\epsilon$. It can be easily seen from the monotonicity of $g(\cdot)$ in Equation \eqref{eq_shortde} that the sequence of constellations $\{\vect{x}^{(\ell)}\}_{\ell = 0}^{\infty}$ are ordered as $\vect{x}^{(\ell)} \succeq \vect{x}^{(\ell + 1)}{\ }\forall{\ }\ell \geq 0$, i.e., $x_i^{(\ell)} \geq x_i^{(\ell + 1)}{\ }\forall{\ }\ell \geq 0, i \in [L]$ (the ordering is pointwise). Since the constellations are all lower bounded by the all-zero constellation $\vect{0}$, the sequence converges pointwise to a limiting constellation $\vect{x}^{(\infty)}$, called the fixed point (FP) of the forward DE.\hfill$\square$
\end{defn}

\vspace*{3mm}It is clear that the FP of forward DE $\vect{x}^{(\infty)}$ satisfies
\[
x_i^{(\infty)} =
\begin{cases}
0, &i \notin [L]\\
\epsilon g(x_{i - \gamma + 1}^{(\infty)}, \cdots, x_{i + \gamma - 1}^{(\infty)}), &i \in [L].
\end{cases}
\]

\begin{defn}[BP Threshold]\vspace*{3mm}
Consider the parallel schedule of the BP forward DE for the $(d_l, d_r, \gamma, L)$ spatially coupled code over a BEC with erasure rate $\epsilon$. The \emp{BP threshold} $\epsilon^\mrm{BP}(d_l, d_r, \gamma, L)$ is defined as the supremum of the channel erasure rates $\epsilon \in [0, 1]$ for which the FP of forward DE is the all-zero constellation, i.e., $\vect{x}^{(\infty)} =  \vect{0}$.\hfill$\square$
\end{defn}

\vspace*{3mm}Table \ref{tab_bpthr} gives the BP thresholds evaluated from BP forward DE for the $(d_l = 3, d_r = 6, \gamma, L)$ coupled ensemble for a few values of $\gamma$ and $L$ rounded to the sixth decimal place.
\begin{table}[!ht]
\centering
\begin{tabular}{cccccc}
\hline
$L\backslash\gamma$ & $2$ & $3$ & $4$ & $5$ & $6$\\
\hline
$16$ & $0.488079$ & $0.488220$ & $0.489806$ & $0.495671$ & $0.505866$ \\
$32$ & $0.488079$ & $0.488150$ & $0.488151$ & $0.488164$ & $0.488294$ \\
$64$ & $0.488078$ & $0.488145$ & $0.488148$ & $0.488149$ & $0.488150$ \\
$128$ & $0.488075$ & $0.488137$ & $0.488142$ & $0.488144$ & $0.488146$ \\
\hline
\end{tabular}
\caption{BP Thresholds $\epsilon^\mrm{BP}(d_l = 3, d_r = 6, \gamma, L)$.}
\label{tab_bpthr}
\end{table}
The MAP threshold of the underlying $(d_l, d_r)$-regular ensemble is $\epsilon^\mrm{MAP}(d_l = 3, d_r = 6) \approx 0.488151$. We see from the table that the BP thresholds for $(d_l, d_r, \gamma, L)$ spatially coupled codes are close to the MAP threshold of the $(d_l, d_r)$-regular unstructured code ensemble. Note that some of the threshold values in Table \ref{tab_bpthr} are larger than the MAP threshold of the underlying $(d_l, d_r)$-regular ensemble. This is because the rates of the spatially coupled ensembles are smaller than the $(d_l, d_r)$-regular ensemble, and depend on the values of $\gamma$ and $L$ as stated in the beginning of this section.

It was shown in \cite{kud_11_tit_ldpccc} that the BP thresholds satisfy
\begin{align}
\lim_{\gamma \tends \infty} \lim_{L \tends \infty} \epsilon^\mrm{BP}(d_l, d_r, \gamma, L) &= \lim_{\gamma \tends \infty} \lim_{L \tends \infty} \epsilon^\mrm{MAP}(d_l, d_r, \gamma, L) \notag \\
&= \epsilon^\mrm{MAP}(d_l, d_r). \notag
\end{align}
This means that the BP threshold \emp{saturates} to the MAP threshold, and we can obtain MAP performance with the reduced complexity of the BP decoder. Later when we analyze the windowed decoder, we will want to keep the coupling length $\gamma$ finite and hence will be concerned with the quantity
\begin{equation} \label{eq_bpthrlinf}
\epsilon^\mrm{BP}(d_l, d_r, \gamma) \triangleq \lim_{L \tends \infty}\epsilon^\mrm{BP}(d_l, d_r, \gamma, L)
\end{equation}
as a measure of the performance of the BP decoder. It immediately follows from \cite[Theorem 10]{kud_11_tit_ldpccc} that
\[
\epsilon^\mrm{BP}(d_l, d_r, \gamma) \le \epsilon^\mrm{MAP}(d_l, d_r).
\]

\section{Windowed Decoding} \label{sec_wd}
The \emp{windowed decoder} (WD) exploits the structure of the spatially coupled codes to break down the BP decoding scheme into a series of sub-optimal decoding steps---we trade-off the performance of the decoder for reduced complexity and decoding latency. When decoding with a window of size $W$, the WD performs BP over the subcode consisting of the first $W$ sections of the variable nodes and their neighboring check nodes and attempts to decode a subset of symbols (those in the first section) within the window. These symbols that we attempt to decode within a window are referred to as the \emp{targeted symbols}. Upon successful decoding of the targeted symbols (or when a maximum number of iterations have been performed) the window slides over one section and performs BP, attempting to decode the targeted symbols in the window in the new position.

More formally, let $\vect{x}$ be the constellation representing the average erasure probability of messages from variables in each of the sections $1$ through $L$. Initially, the window consists only of the first $W$ sections in the chain. We will refer to this as the \emph{first window configuration}, and as the window slides to the right, we will increment the window configuration. In other words, when the window has slid through $(c - 1)$ sections to the right (when it consists of sections $c, c + 1, \cdots, c + W - 1$), it is said to be in the $c^\text{th}$ window configuration. The $c^\text{th}$ \emph{window constellation}, denoted $\vect{y}_{\{c\}}$, is the average erasure probability of the variables in the $c^\text{th}$ window configuration. Thus,
\[
\vect{y}_{\{c\}} = (y_{1, \{c\}}, y_{2, \{c\}}, \cdots, y_{W, \{c\}}) = (x_{c}, x_{c + 1}, \cdots, x_{c + W - 1})
\]
for $c \in [L]$, where we assume that $x_c = 0{\ }\forall{\ }c > L$. Thus the $c^\text{th}$ window constellation, $\underline{y}_{\{c\}}$, represents the ``active'' sections within the constellation $\underline{x}$. While referring to the entire constellation after the action of the $c^\text{th}$ window, we will write $\underline{x}_{\{c\}}$. When the window configuration being considered is clear from the context, with some abuse of notation, we drop the $\{c\}$ from the notation and write $\vect{y} = (y_1, \cdots, y_W)$ to denote the window constellation.

\begin{rem}[Note on notation]\vspace*{3mm}
When we wish to emphasize the size of the window when we write the constellation, we write $\vect{y}_{\langle W\rangle} = (y_{i, \langle W\rangle}, y_{2, \langle W\rangle}, \cdots, y_{W, \langle W\rangle})$. Note that the window configuration and the window size are specified as subscripts within curly brackets $\{\cdot\}$ and angle brackets $\langle\cdot\rangle$, respectively. Finally, when the constellation after a particular number of iterations $\ell$ of DE is to be specified, we write $\vect{y}^{(\ell)} = (y_1^{(\ell)}, y_2^{(\ell)}, \cdots, y_W^{(\ell)})$, where the iteration number appears as a superscript within parentheses $(\cdot)$. Although $\vect{y}_{\{c\}, \langle W\rangle}^{(\ell)}$ would be the most general way of specifying the window constellation for the $c^\text{th}$ window configuration with a window of size $W$ after $\ell$ iterations of DE, for notational convenience we will write as few of these parameters as possible based on the relevance to the discussion.\hfill$\square$
\end{rem}

\subsection{Complexity and Latency} \label{ssec_complat}
For the BP decoder, the number of iterations required to decode all the symbols in a $(d_l, d_r, \gamma, L)$ spatially coupled code depends on the channel erasure rate $\epsilon$. Whereas when $\epsilon \in [0, \epsilon^\mathrm{BP}(d_l, d_r)]$ this required number of iterations can be fixed to a constant number, when $\epsilon \in (\epsilon^\mathrm{BP}(d_l, d_r), \epsilon^\mathrm{BP}(d_l, d_r, \gamma, L)]$ the number of iterations scales as $O(L)$ \cite{olm_10_pvt_finlenscc}. Therefore, in the waterfall region, the complexity of the BP decoder scales as $O(ML^2)$. For the WD of size $W$, if we let the number of iterations performed scale as $O(W)$, the overall complexity is of the order $O(MW^2L)$. Thus, for small window sizes $W < \sqrt{L}$, we see that the complexity of the decoder can be reduced. A larger reduction in the complexity is possible if we fix the number of iterations performed within each window.

Another advantage of using the WD is that the decoder only needs to know the symbols in the first $W$ sections of the code to be able to decode the targeted symbols. Therefore, in latency-constrained applications, the decoder can work on-the-fly, resulting in a latency which is a fraction $\frac{W}{L}$ that of the BP decoder.

\subsection{Asymptotic Performance} \label{ssec_asymperf}
The asymptotic performance of the $(d_l, d_r, \gamma, L)$ spatially coupled ensemble with WD can be analyzed using density evolution as was done for the BP decoder. We will consider the performance of the ensemble with $M \rightarrow \infty$ when the transmission happens over a BEC with channel erasure rate $\epsilon \in [0, 1]$. Further we will assume that for each window configuration, infinite rounds of message passing are performed.

\begin{defn}[WD Forward Density Evolution] \label{defn_wdfde}\vspace*{3mm}
Consider the WD of a $(d_l, d_r, \gamma, L)$ spatially coupled code over a BEC with channel erasure rate $\epsilon$ with a window of size $W$. We can write the forward DE equation as follows. Set the initial constellation $\vect{x}_{\{0\}}$ according to
\[
x_{i, \{0\}} = 
\begin{cases}
1, &i \in [L]\\
0, &i \notin [L].
\end{cases}
\]
For every window configuration $c = 1, 2, \cdots, L$, let
\[
\vect{y}_{\{c\}}^{(0)} = (x_{c, \{c - 1\}}, x_{c + 1, \{c - 1\}}, \cdots, x_{c + W - 1, \{c - 1\}})
\]
and evaluate the sequence of window constellations $\{\vect{y}_{\{c\}}^{(\ell)}\}_{\ell = 1}^{\infty}$ using the update rule
\[
y_{i, \{c\}}^{(\ell)} = \epsilon g(y_{i - \gamma + 1, \{c\}}^{(\ell - 1)}, \cdots, y_{i + \gamma - 1, \{c\}}^{(\ell - 1)}), i \in [W],
\]
where for every $\ell$, for indices $i \notin [W]$, we set $y_{i, \{c\}}^{(\ell)} = x_{c + i - 1, \{c - 1\}}$ on the right hand side. 
Then set $\vect{x}_{\{c\}}$ as
\[\hspace*{22.5mm}
x_{i, \{c\}} =
\begin{cases}
x_{i, \{c - 1\}}, &i \neq c\\
y_{1, \{c\}}^{(\infty)}, &i = c.\hspace*{23mm}\square
\end{cases}
\]
\end{defn}
\vspace*{3mm}
\begin{dis}
Note that the constellation $\vect{x}_{\{c\}}$ keeps track of the erasure probabilities of targeted symbols of all window configurations up to the $c^\text{th}$, followed by erasure probability of $1$ for the variables in sections $c + 1$ through $L$, and zeros for sections outside this range. As defined, $\vect{x}_{\{c\}}$ discards all information obtained by running the WD in its $c^\text{th}$ configuration apart from the values corresponding to the targeted symbols. In practice, it is more efficient to define
\[
x_{i, \{c\}} =
\begin{cases}
x_{i, \{c - 1\}}, &i \notin \{c, c + 1, \cdots, c + W - 1\}\\
y_{i - c + 1, \{c\}}^{(\infty)}, &\text{otherwise}.
\end{cases}
\]
In the sequel, we will stick to Definition \ref{defn_wdfde}. We do this for two reasons: first, discarding some information between two window configurations can only perform worse than retaining all the information; and second, this assumption makes the analysis simpler since we then have $\vect{y}_{\{c\}}^{(0)} = \vect{1}{\ }\forall{\ }c \in [L]$.\hfill$\square$\vspace*{3mm}
\end{dis}

Definition \ref{defn_wdfde} implicitly assumes that the limiting window constellations $\vect{y}_{\{c\}}^{(\infty)}$ exist. The following guarantees that the updates for $x_{i, \{c\}}$ are well-defined.

\begin{lem}[$c^\text{th}$ Window Configuration FP of FDE]\vspace*{3mm}
Consider the WD forward DE (FDE) of a $(d_l, d_r, \gamma, L)$ spatially coupled code over a BEC with erasure rate $\epsilon$ with a window of size $W$. Then the limiting window constellation $\vect{y}_{\{c\}}^{(\infty)}$ exists for each $c \in [L]$. We refer to this constellation as the $c^\text{th}$ window configuration FP of forward DE.
\end{lem}
\begin{IEEEproof}
As noted earlier, $\vect{y}_{\{c\}}^{(0)} = \vect{1}{\ }\forall{\ }c \in [L]$, and $\vect{y}_{\{c\}}^{(0)} = \vect{1} \succeq \vect{\epsilon} \succeq \vect{y}_{\{c\}}^{(1)}$. By induction, from the monotonicity of $g(\cdot)$, this implies that $\vect{y}_{\{c\}}^{(\ell)} \succeq \vect{y}_{\{c\}}^{(\ell + 1)}{\ }\forall{\ }\ell \geq 0$. Since these constellations are lower bounded by $\vect{0}$, the $c^\text{th}$ window configuration FP of FDE $\vect{y}_{\{c\}}^{(\infty)}$ exists for every $c \in [L]$.
\end{IEEEproof}

\vspace*{3mm}The $c^\text{th}$ window configuration FP of forward DE therefore satisfies
\begin{equation} \label{eq_cwcfpfde}
y_{i, \{c\}}^{(\infty)} =
\begin{cases}
x_{c + i - 1, \{c - 1\}}, &i \notin [W]\\
\epsilon g(y_{i - \gamma + 1, \{c\}}^{(\infty)}, \cdots, y_{i + \gamma - 1, \{c\}}^{(\infty)}), &i \in [W]
\end{cases}
\end{equation}
for every $c \in [L]$. Since the $\vect{x}_{\{0\}}$ vector has non-zero values by definition, from the continuity of the WD FDE equations, so do the vectors $\vect{x}_{\{c\}}{\ }\forall{\ }c$. Hence $\vect{0}$ cannot satisfy Equation \eqref{eq_cwcfpfde}, i.e.,  $\vect{0}$ cannot be the $c^\text{th}$ window configuration FP of forward DE. Therefore, $\vect{y}_{\{c\}}^{(\infty)} \succ \vect{0}{\ }\forall{\ }c \in [L]$. This means that WD can never reduce the erasure probability of the symbols of a spatially coupled code to zero, although it can be made arbitrarily small by using a large enough window. Therefore, an acceptable \emph{target erasure rate} $\delta$ forms a part of the description of the WD. We say that the WD is successful when $\vect{x}_{\{L\}} \preceq \vect{\delta}$.

\begin{lem}[Maximum of $\vect{x}_{\{L\}}$] \label{lem_maxxl} \vspace*{3mm}
The vector $\vect{x}_{\{L\}}$ obtained at the end of WD forward DE satisfies $x_{i - 1, \{L\}} \leq x_{i, \{L\}}{\ }\forall{\ }i \in [L - W - \gamma + 2]$. Moreover, $\exists{\ }\hat{x} \in [0, 1]$ independent of $L$ such that $x_{i, \{L\}} \leq \hat{x}{\ }\forall{\ }i$.
\end{lem}
\begin{proof}
By definition, $x_{i, \{L\}} = y_{1, \{i\}}^{(\infty)}{\ }\forall{\ }i \in [L]$. The claim is true for $i = 1$ since $x_{1, \{L\}} = y_{1, \{1\}}^{(\infty)} \geq 0 = x_{0, \{L\}}$. For the $i^\text{th}$ window configuration, it is clear from Definition \ref{defn_wdfde} that $\vect{y}_{\{i - 1\}}^{(0)} \preceq \vect{y}_{\{i\}}^{(0)}, i \in [L - W - \gamma + 2]$. By induction, from the monotonicity of $g(\cdot)$, it follows that $y_{1, \{i - 1\}}^{(\infty)} \leq y_{1, \{i\}}^{(\infty)}$ for $i$ in this range.

For $i > L - W - \gamma + 2$, the above claim is not valid because we defined $x_{j, \{c\}} = 0$ for $j > L$ and we cannot make use of the monotonicity of $g(\cdot)$ since some arguments (corresponding to sections up to the $L^\text{th}$ section) are increasing and others (corresponding to the sections beyond the $L^\text{th}$ section) decreasing. Nevertheless, we can still claim that $x_{i, \{L\}} \leq x_{i, \{\infty\}}{\ }\forall{\ }i \in \mathbb{N}$ where $\vect{x}_{\{\infty\}}$ is the vector of erasure probabilities obtained after WD for a spatially coupled code with an infinite chain length, i.e., $L = \infty$. For $L = \infty$, the sequence $\{x_{i, \{\infty\}}\}$ is non-decreasing and since the $x_{i, \{\infty\}}$ are probabilities, they are in the bounded, closed interval $[0, 1]$. Consequently, the limit $\lim_{i \rightarrow \infty} x_{i, \{\infty\}}$ exists in the interval $[0, 1]$, and $\lim_{i \rightarrow \infty} x_{i, \{\infty\}} = \sup_i x_{i, \{\infty\}} \triangleq \hat{x}$.\vspace*{3mm}
\end{proof}

As a consequence of Lemma \ref{lem_maxxl}, we can say that the WD is successful when $\hat{x} \leq \delta$. This definition of the success of WD is independent of the chain length $L$ and allows us to compare the performance of WD to that of the BP decoder through the thresholds defined in Equation \eqref{eq_bpthrlinf}. Note that although the upper bound for $\hat{x}$ in Lemma \ref{lem_maxxl} is a trivial bound, we will in the following give conditions when $\hat{x}$ can be made smaller than an arbitrarily chosen $\delta$, thereby characterizing the WD thresholds.

\begin{defn}[WD Thresholds] \label{defn_wthr} \vspace*{3mm}
Consider the WD of a $(d_l, d_r, \gamma, L)$ spatially coupled code over a BEC of erasure rate $\epsilon$ with a window of size $W$. The WD threshold $\epsilon^\mathrm{WD}(d_l, d_r, \gamma, W, \delta)$ is defined as the supremum of channel erasure rates $\epsilon$ for which $\hat{x} \leq \delta$.\hfill$\square$\vspace*{3mm}
\end{defn}
\begin{dis}
Since we defined the WD threshold based on $\hat{x}$, it is clear that this is independent of the chain length $L$. On the other hand, if we used $\max_{i \in [L]} x_{i, \{L\}} \leq \delta$ as the condition for success of the WD in our definition, we would obtain an $L$-dependent threshold. But $\hat{x}$ denotes the ``worst-case'' remanant erasure probability after WD, and imposing constraints on $\hat{x}$ therefore guarantees good performance for codes with any $L$.

Note that keeping $\hat{x} \leq \delta$ is sufficient to guarantee an {\it a-posteriori} erasure probability $p_e$ smaller than $\delta$ because
\[
\hspace*{14.5mm}p_e = \epsilon \Big(\frac{\hat{x}}{\epsilon}\Big)^{\frac{d_l}{d_l - 1}} = \hat{x} \Big(\frac{\hat{x}}{\epsilon}\Big)^{\frac{1}{d_l - 1}} \leq \hat{x} \leq \delta.\hspace*{14.5mm}\square
\]
\end{dis}

\vspace*{3mm}We will now state the main result in this paper and prove it in the following section.
\begin{thm}[WD Threshold Bound] \label{thm_wdthr} \vspace*{3mm}
Consider windowed decoding of the $(d_l, d_r, \gamma, L)$ spatially coupled ensemble over the binary erasure channel. Then for a target erasure rate $\delta < \delta_*$, there exists a positive integer $W_{\min}(\delta)$ such that when the window size $W \ge W_{\min}(\delta)$ the WD threshold satisfies
\begin{align}
\epsilon^\mrm{WD}(&d_l, d_r, \gamma, W, \delta) \ge \Big(1 - \frac{d_ld_r}{2}\delta^{\frac{d_l - 2}{d_l - 1}}\Big) \notag \\
&\times \Big(\epsilon^\mrm{BP}(d_l, d_r, \gamma) - e^{-\frac{1}{\mathsf{B}}(\frac{W}{\gamma - 1} - \mathsf{A}\ln\ln\frac{\mathsf{D}}{\delta} - \mathsf{C})}\Big). \label{eq_wdthrthm}
\end{align}
Here $\mathsf{A}, \mathsf{B}, \mathsf{C}, \mathsf{D}$ and $\delta_*$ are strictly positive constants that depend only on the ensemble parameters $d_l, d_r$ and $\gamma$.\hfill$\blacksquare$
\end{thm}

Theorem \ref{thm_wdthr} says that the WD thresholds approach the BP threshold $\epsilon^\mathrm{BP}(d_l, d_r, \gamma)$ defined in Equation \eqref{eq_bpthrlinf} at least exponentially fast in the ratio of the size of the window $W$ to the coupling length $\gamma$ for a fixed target erasure probability $\delta < \delta_*$. Moreover, the sensitivity of the bound to changes in $\delta$ is small in the exponent in \eqref{eq_wdthrthm} owing to the $\log\log\frac{1}{\delta}$ factor, but larger in the first term in the product on the right hand side of \eqref{eq_wdthrthm} where it is roughly linear in $\delta$. However, since we intend to set $\delta$ to be very small, e.g. $10^{-15}$, the first term does not influence the bound heavily. The requirement that $W \geq W_{\min}(\delta)$ is necessary to keep the term within parentheses in the exponent non-negative. Therefore the minimum window size required, $W_{\min}(\delta)$, also depends on the constants $\mathsf{A}, \mathsf{C}$ and $\mathsf{D}$ and, in turn, on the ensemble parameters $d_l, d_r$ and $\gamma$.

The bound guaranteed by Theorem \ref{thm_wdthr} is actually fairly loose. Numerical results suggest that the minimum window size $W_{\min}(\delta)$ is actually much smaller than the bound obtained from analysis (cf. Section \ref{sec_an}). Density evolution also reveals that for a fixed window size, the WD thresholds are much closer to the BP threshold than the bound obtained from Theorem \ref{thm_wdthr}. 

We note here that the gap between analytical results and numerical experiments is mainly due to the reliance on bounding the density evolution function in Equation \eqref{eq_shortde} using the counterpart for regular unstructured LDPC ensembles, which proves to be easier to handle than the multivariate Equation \eqref{eq_shortde} (See, e.g., the bound in \ref{eq_proofi}). However, the scaling of the WD thresholds with the window size and the target erasure probability seem to be as dictated by the bound in \eqref{eq_wdthrthm}, suggesting that Theorem \ref{thm_wdthr} captures the essence of the WD algorithm.

Table \ref{tab_wdthr} gives the WD thresholds obtained through forward DE for the $(d_l = 3, d_r = 6, \gamma = 3, L)$ spatially coupled ensemble for different target erasure rates $\delta$ and different window sizes $W$. These thresholds have been rounded to the sixth decimal point.
\begin{table}[!ht]
\centering
\begin{tabular}{cccc}
\hline
$W\backslash\delta$ & $10^{-6}$ & $10^{-12}$ & $10^{-18}$ \\
\hline
$4$ &	$0.068403$ &	$0.000772$ &	$0.000008$ \\
$8$ &	$0.472992$ &	$0.390749$ &	$0.254339$ \\	
$16$ &	$0.487504$ &	$0.487504$ &	$0.487504$ \\
\hline
\end{tabular}
\caption{WD Thresholds $\epsilon^\mrm{WD}(d_l = 3, d_r = 6, \gamma = 3, W, \delta)$.}
\label{tab_wdthr}
\end{table}
A few comments are in order. As can be seen from the table, the thresholds are close to $\epsilon^\mathrm{BP}(d_l = 3, d_r = 6, \gamma = 3) \approx 0.488137$ even for window sizes that are much smaller than the $W_{\min}(\delta)$ obtained analytically, e.g., $W = 16$. Moreover, the WD thresholds are more sensitive to changes in $\delta$ for small window sizes where the bound in Theorem \ref{thm_wdthr} is not valid. It is obvious that the thresholds decrease as $\delta$ is decreased. Also note that for a fixed target erasure rate, the window size can be made large enough to make the WD thresholds close to the BP threshold.

\section{Performance Analysis} \label{sec_an}
In this section, we prove Theorem \ref{thm_wdthr} in steps. First, we analyze the performance of the first window configuration. We will characterize the first window configuration FP of forward DE. We will establish that for the variables in the first section of the window, the FP erasure probability can be made small at least double-exponentially in the size of the window. We will show that this is possible for all channel erasure rates smaller than a certain $\epsilon$, which we will call the \emph{first window threshold} $\epsilon^\mathrm{FW}(d_l, d_r, \gamma, W, \delta)$, provided the window size is larger than a certain minimum size.

Once we have this, we consider the performance of the $c^\text{th}$ window configuration for $1 < c \leq L$. In this case also, we will show that the FP erasure probability of the first section within the window is guaranteed to decay double-exponentially in the window size. As for the first window configuration, this result holds provided the window size is larger than a certain minimal size and this time the minimal size is slightly larger than the minimal size required for the first window configuration. Moreover, such a result is true for channel erasure rates smaller than a value which is itself smaller than the first window threshold, and this value will be our lower bound for the WD threshold.

\subsection{First Window Configuration} \label{ssec_fw}
From Definition \ref{defn_wdfde}, forward DE for the first window configuration amounts to the following. Set $\vect{y}_{\{1\}}^{(0)} = \vect{1}$ and evaluate the sequence of window constellations $\{\vect{y}_{\{1\}}^{(\ell)}\}_{\ell = 1}^{\infty}$ according to
\begin{equation} \label{eq_wcfde}
y_{i, \{1\}}^{(\ell)} =
\begin{cases}
0, & i \leq 0\\
\epsilon g(y_{i - \gamma + 1, \{1\}}^{(\ell - 1)}, \cdots, y_{i + \gamma - 1, \{1\}}^{(\ell - 1)}), & i \in [W]\\
1, & i > W.
\end{cases}
\end{equation}
Since $\vect{y}_{\{1\}}^{(0)}$ is non-decreasing, i.e., $y_{i, \{1\}}^{(0)} \leq y_{i + 1, \{1\}}^{(0)}{\ }\forall{\ }i$, so is the first window configuration FP, $\vect{y}_{\{1\}}^{(\infty)}$, by induction and monotonicity of $g(\cdot)$.

Fig. \ref{fig_wdfp} shows the first window configuration FP of forward DE for the $(d_l = 3, d_r = 6, \gamma = 3, L)$ ensemble with a window of size $W = 16$ for a channel erasure rate $\epsilon = 0.48812$.
\begin{figure}[!ht]
\centering
\figselect{\include{wcfpfde}}{\includegraphics{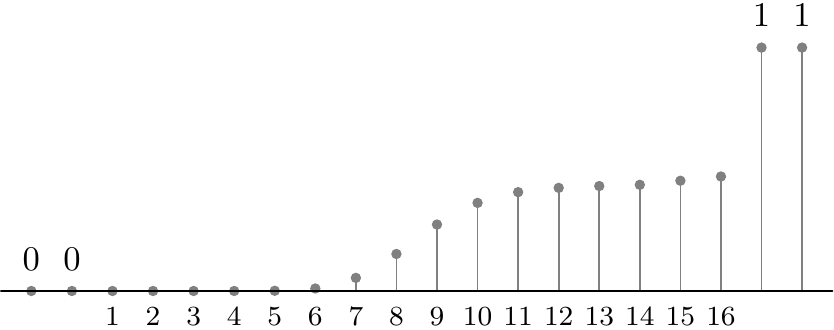}}
\caption{The first window configuration FP of forward DE for the $(d_l = 3, d_r = 6, \gamma = 3, L)$ ensemble with a window of size $W = 16$ for $\epsilon = 0.48812$. The left and the right boundaries are fixed at $0$ and $1$ respectively. The sections within the window are indexed from $1$ to $W = 16$. The first section has a FP erasure probability $y_{1, \{1\}}^{(\infty)} \approx 2\times 10^{-15}$.}
\label{fig_wdfp}
\end{figure}

The scheduling scheme used in the definition of the window configuration FPs is what is called the \emph{parallel} schedule. In general, we can consider a scheduling scheme where, in each step, a subset of the sections within the window are updated. We say that such an arbitrary scheduling scheme is admissible if every section is updated infinitely often with the correct boundary conditions, i.e., with the correct values set at the left and the right ends of the window. It is easy to see from the standard argument of nested computation trees (see, e.g., \cite{kud_11_tit_ldpccc}) that the FP is independent of the scheduling scheme.

We know that the first window configuration FP of forward DE, $\vect{y}_{\{1\}}^{(\infty)}$, is non-decreasing, i.e., $y_{i, \{1\}}^{(\infty)} \leq y_{i + 1, \{1\}}^{(\infty)}{\ }\forall{\ }i$. The following shows the ordering of the FP values of individual sections in windows of different sizes. With the understanding that we are considering only the first window configuration in this subsection, we will drop the window configuration number from the notation for window constellations thoughout this subsection for convenience.

\begin{lem}[FPs and Window Size] \label{lem_fpord} \vspace*{3mm}
Let $\vect{y}_{\langle W\rangle}$ and $\vect{y}_{\langle W + 1\rangle}$ denote the first window configuration FPs of forward DE with windows of sizes $W$ and $W + 1$ respectively for $\epsilon \in [0, 1]$. Then,
\[
y_{i, \langle W\rangle} \ge y_{i, \langle W + 1\rangle} \ge y_{i - 1, \langle W\rangle}
\]
where $y_{i, \langle W\rangle}$ denotes the FP erasure probability of the $i^\text{th}$ section in a window of size $W$.
\end{lem}
\begin{proof}
Consider the following schedule. Set $\vect{y}_{\langle W + 1\rangle}^{(0)} = (\vect{y}_{\langle W\rangle}, 1)$ and evaluate the sequence of window constellations $\{\vect{y}_{\langle W + 1\rangle}^{(\ell)}\}$ according to Equation \eqref{eq_wcfde}. Clearly, we have
\[
\vect{y}_{\langle W + 1\rangle}^{(1)} \preceq (\vect{y}_{\langle W\rangle}, \epsilon) \preceq (\vect{y}_{\langle W\rangle}, 1) = \vect{y}_{\langle W + 1\rangle}^{(0)}
\]
so that the sequence $\{\vect{y}_{\langle W + 1\rangle}^{(\ell)}\}$ is pointwise non-increasing by induction. We claim that this schedule is admissible. This is true because the DE updates are first performed infinitely many times over the first $W$ sections to obtain $\vect{y}_{\langle W + 1\rangle}^{(0)}$, and then over all the $W + 1$ sections infinitely many times again. Therefore the updates are performed over all sections infinitely often with the correct boundary conditions. The limiting FP must hence be exactly $\vect{y}_{\langle W + 1\rangle}$ and the first inequality in the statement of the lemma holds. Intuitively, this is true because in going from $W$ to $W + 1$ and checking the $i^\text{th}$ section, we have moved further away from the right end of the window (where $y_i = 1$) while remaining at the same distance from the left end (where $y_i = 0$).

To prove the second inequality, consider the following schedule. Set $\vect{y}_{i, \langle W\rangle}^{(0)} = \vect{y}_{i + 1, \langle W + 1\rangle}, i = 1, \cdots, W$ and evaluate the sequence of constellations $\{\vect{y}_{\langle W\rangle}^{(\ell)}\}$ according to Equation \eqref{eq_wcfde}. Since $\vect{y}_{0, \langle W\rangle}^{(0)} = 0 \leq \vect{y}_{1, \langle W + 1\rangle}$, we must have $\vect{y}_{\langle W\rangle}^{(1)} \preceq \vect{y}_{\langle W\rangle}^{(0)}$ and by induction the sequence of constellations thus obtained is also pointwise non-increasing. Again we claim that the above mentioned schedule is admissible. This is true because we first update all $W$ sections within the window and also the zeroth section infinitely often, and then set the boundary condition that the zeroth section also has all variables completely known. In all, every section within the window gets updated infinitely often with the correct boundary conditions. The limiting FP must hence be exactly $\vect{y}_{\langle W\rangle}$ and the second inequality claimed in the statement of the lemma follows. As in the previous case, this is intuitively true because in going from the $(i + 1)^\text{th}$ section with window size $W + 1$ to the $i^\text{th}$ section with window size $W$, we have moved closer to the left end of the window while maintaining the distance from the right end. \vspace*{3mm}
\end{proof}

We now give some bounds on the FP erasure probabilities of individual sections within a window.

\begin{lem} [Bounds on FP] \label{lem_boundsfp} \vspace*{3mm}
Consider the WD of the $(d_l, d_r, \gamma, L)$ ensemble with a window of size $W$ over a channel with erasure rate $\epsilon$ and $d_l \geq 3$. The first window configuration FP $\vect{y}$ satisfies
\begin{align}
y_i &\ge \Big( \epsilon (\frac{\gamma - 1}{2\gamma})^{d_l - 1} \Big) ^ {\frac{(d_l - 1) ^ j - 1}{d_l - 2}}y_{i + j}^{(d_l - 1)^j} \notag \\
y_i &\le \epsilon \Big(1 - \alpha_k(1 - y_{i + k})^{d_r - 1} \Big)^{d_l - 1} \notag
\end{align}
for $i \in [1, W], j \in [0, W + 1 - i], k \in [0, \gamma - 1]$, where $\alpha_k = (1 - \frac{(\gamma - k - 1)(\gamma - k)}{2\gamma^2})^{d_r - 1}$. \hfill$\blacksquare$
\end{lem}

\vspace*{3mm}We relegate the proof to Appendix \ref{app_boundsfp}. The following shows that once the FP erasure probability of a section within the window is smaller than a certain value, it decays very quickly as we move further to the left in the window.

\begin{lem}[Doubly-Exponential Tail of the FP] \label{lem_detail} \vspace*{3mm}
Consider WD of the $(d_l, d_r, \gamma, L)$ ensemble with a window of size $W$ over a channel with erasure rate $\epsilon \in (0, 1)$. Let $d_l \ge 3$ and let $\vect{y}$ be the first window configuration FP of forward DE. If there exists an $i \in [W]$ such that $y_i < \delta_0 \triangleq \Big((d_r - 1)^{\frac{d_l - 1}{d_l - 2}}\Big)^{-1}$, then
\[
y_{i - j(\gamma - 1)} \le \Psi e^{-\psi(d_l - 1)^j}
\]
where $\Psi = \delta_0\epsilon^{\frac{-1}{d_l - 2}}$ and $\psi = \ln (\frac{\Psi}{\delta_0}) = \frac{1}{d_l - 2}\ln\frac{1}{\epsilon} > 0$.
\end{lem}
\begin{proof}
Since the FP is non-decreasing, we have
\begin{align}
y_{i - (\gamma - 1)} &= \epsilon g(y_{i - 2(\gamma - 1)}, \cdots, y_i) \notag \\
&\le \epsilon g(y_i, y_i, \cdots, y_i) \notag \\
&= \epsilon(1 - (1 - y_i)^{d_r - 1})^{d_l - 1} \label{eq_dldreq} \\
&\le \epsilon((d_r - 1)y_i)^{d_l - 1} \notag
\end{align}
which can be applied recursively to obtain
\begin{align}
y_{i - j(\gamma - 1)} &\le \epsilon^{\frac{(d_l - 1)^j - 1}{d_l - 2}}(d_r - 1)^{\frac{d_l - 1}{d_l - 2}((d_l - 1)^j - 1)} y_i^{(d_l - 1)^j} \notag \\
&< \epsilon^{\frac{(d_l - 1)^j - 1}{d_l - 2}}(d_r - 1)^{\frac{d_l - 1}{d_l - 2}((d_l - 1)^j - 1)} \delta_0^{(d_l - 1)^j} \notag \\
&\triangleq \Psi e^{-\psi(d_l - 1)^j} \label{eq_detail}
\end{align}
where $\Psi$ and $\psi$ are as defined in the statement. It is worthwhile to note that $\delta_0$ is a lower bound on the \emp{breakout value} for the $(d_l, d_r)$-regular ensemble \cite{len_05_tit_brkout}. The emergence of the breakout value in this context is not entirely unexpected since it is known that for the $(d_l, d_r)$-regular ensemble, the erasure probability decays double-exponentially in the number of iterations below the breakout value, and in case of spatially coupled ensembles, the counterpart for the number of iterations is the number of sections (cf. Equation \eqref{eq_dldreq}).
\end{proof}

We now show that the FP erasure probability of a message from a variable node in the first section, $y_1$, can be made small by increasing the window size $W$ for any $\epsilon < \epsilon^\mrm{BP}(d_l, d_r, \gamma)$. Assuming that the window size is ``large enough,'' we will count the number of sections, starting from the right, that have a FP erasure probability larger than a small $\delta$ for a channel erasure rate $\epsilon = \epsilon^\mrm{BP}(d_l, d_r, \gamma) - \Delta\epsilon$.
\begin{defn}[Transition Width] \label{defn_tranwid} \vspace*{3mm}
Consider WD of a $(d_l, d_r, \gamma, L)$ spatially coupled code over a BEC of erasure rate $\epsilon$. Let $\vect{y}$ be the $1^\text{st}$ window configuration FP of forward DE. Then we define the transition width $\tau(\epsilon, \delta)$ of $\vect{y}$ as
\[
\hspace*{17mm}\tau(\epsilon, \delta) = |\{i \in [W] : \delta < y_i \leq 1\}|.\hspace*{18mm}\square
\]
\end{defn}

\vspace*{3mm}Note that from the definition of the transition width, it depends on the window size $W$. We first upper bound $\tau(\epsilon, \delta) \leq \hat{\tau}(\epsilon, \delta)$ such that the upper bound is independent of the window size $W$, and then claim from Lemma \ref{lem_fpord} that by employing a window whose size is larger than $\hat{\tau}(\epsilon, \delta)$, we can guarantee $y_1 \le \delta$.

\begin{defn}[First Window Threshold] \label{defn_fwthr} \vspace*{3mm}
Consider WD of the $(d_l, d_r, \gamma, L)$ spatially coupled ensemble with a window of size $W$ over a BEC with eraure rate $\epsilon$. The first window threshold $\epsilon^\mathrm{FW}(d_l, d_r, \gamma, W, \delta)$ is defined as the supremum of channel erasure rates for which the first window configuration FP of forward DE $\vect{y}$ satisfies $y_1 \leq \delta$.\hfill$\square$
\end{defn}
\vspace*{3mm}From Definitions \ref{defn_tranwid} and \ref{defn_fwthr}, we can see that by ensuring that $W > \hat{\tau}(\epsilon, \delta)$, we can bound $\epsilon^\mrm{FW}(d_l, d_r, \gamma, W, \delta) \ge \epsilon$.

\begin{prop}[Maximum Transition Width] \label{prop_tw} \vspace*{3mm}
Consider the first window configuration FP of forward DE $\vect{y}$ for the $(d_l, d_r, \gamma, L)$ spatially coupled ensemble with a window of size $W < L$ for $\epsilon \in [\frac{\epsilon^\mrm{BP}(d_l, d_r, \gamma) + \epsilon^\mrm{BP}(d_l, d_r)}{2}, \epsilon^\mrm{BP}(d_l, d_r, \gamma)) = \msc{E}$. Then,
\[
\tau(\epsilon, \delta) \le (\gamma - 1)\Big(\mathsf{A}\ln\ln\frac{\mathsf{D}}{\delta} + \mathsf{B}\ln\frac{1}{\Delta\epsilon} + \hat{\mathsf{C}} \Big) \triangleq \hat{\tau}(\epsilon, \delta)
\]
provided $\delta \le \delta_0$. Here $\Delta\epsilon = \epsilon^\mrm{BP}(d_l, d_r, \gamma) - \epsilon$, and $\mathsf{A}, \mathsf{B}, \hat{\mathsf{C}}, \mathsf{D}$ and $\delta_0$ are strictly positive constants that depend only on the ensemble parameters $d_l, d_r$ and $\gamma$.\hfill$\blacksquare$\vspace*{3mm}
\end{prop}

The proof is given in Appendix \ref{app_tw}. This means that the smallest window size that guarantees $y_1 \le \delta$ for a channel erasure rate $\frac{\epsilon^\mrm{BP}(d_l, d_r, \gamma) + \epsilon^\mrm{BP}(d_l, d_r)}{2}$ is
\begin{align}
\hat{W}_{\min}(\delta) &= \Big\lfloor(\gamma - 1)\Big(\mathsf{A}\ln\ln\frac{\mathsf{D}}{\delta} + \mathsf{B}\ln\frac{1}{\Delta\epsilon_{\max}} + \hat{\mathsf{C}}\Big)\Big\rfloor + 1 \notag \\
&> \hat{\tau} \Big(\frac{\epsilon^\mathrm{BP}(d_l, d_r, \gamma) + \epsilon^\mathrm{BP}(d_l, d_r)}{2}, \delta \Big) \notag
\end{align}
where $\Delta\epsilon_{\max} = \frac{\epsilon^\mrm{BP}(d_l, d_r, \gamma) - \epsilon^\mrm{BP}(d_l, d_r)}{2}$. When $W \ge \hat{W}_{\min}(\delta)$, we have
\begin{align} \label{eq_fwthr}
\epsilon^\mrm{FW}(d_l, d_r, \gamma, W, \delta) &\ge \epsilon^\mrm{BP}(d_l, d_r, \gamma) \notag \\
&\hspace*{3mm}- e^{-\frac{1}{\mathsf{B}}(\frac{W}{\gamma - 1} - \mathsf{A}\ln\ln\frac{\mathsf{D}}{\delta} - \hat{\mathsf{C}})}.
\end{align}

\begin{dis} \vspace*{3mm}
We restricted $\epsilon \in \mathscr{E}$ in Proposition \ref{prop_tw} to obtain constants that are independent of $\epsilon$. As can be seen from the proof of the proposition, these constants are dependent on $\epsilon$, unless each is optimized in the range $\mathscr{E}$. As we let the minimum $\epsilon$ in $\mathscr{E}$ approach $\epsilon^\mathrm{BP}(d_l, d_r)$, the constants in the expression for $\hat{\tau}(\epsilon, \delta)$ blow up and the upper bound will be useless. It is therefore necessary to keep the minimum of $\epsilon \in \mathscr{E}$ strictly larger than $\epsilon^\mathrm{BP}(d_l, d_r)$ and the value chosen in the above was motivated by our intent to ensure that the first window threshold was closer to $\epsilon^\mathrm{BP}(d_l, d_r, \gamma)$ than to $\epsilon^\mathrm{BP}(d_l, d_r)$. Note that the increase in the upper bound for $\tau(\epsilon, \delta)$ with decrease in $\epsilon$ is purely an artifact of the upper bounding technique we have employed; i.e., it is obvious that as we decrease $\epsilon$, $\tau(\epsilon, \delta)$ also decreases.\hfill$\square$
\end{dis}

\subsection{$c^\text{th}$ Window Configuration, $1 < c \leq L$} \label{ssec_sw}
We now evaluate the performance of the windowed decoding scheme when the window has slid certain number of sections from the left end of the code. We arrive at conditions under which $\hat{x}$ is guaranteed to be smaller than $\delta$ while operating with a window of size $W$. We start by establishing a property of $\hat{x}$.
\begin{lem}[FP Equation Involving $\hat{x}$] \label{lem_hatx} \vspace*{3mm}
Consider the function $\Omega(\vect{y})$ where
\[
\Omega(y_i) =
\begin{cases}
\pi(y_i), & i < 1\\
\epsilon g(\pi(y_{i - \gamma + 1}), \cdots, \pi(y_{i + \gamma - 1})), & i \in [W]\\
1, & i > W
\end{cases}
\]
where
\[
\pi(y_i) =
\begin{cases}
y_1, & i < 1\\
y_i, & i \in [W]\\
1, & i > W.
\end{cases}
\]
Then there exists a solution $\vect{\omega}$ to the equation $\vect{y} = \Omega(\vect{y})$ such that $\omega_1 = \hat{x}$. Moreover, $\vect{\omega}$ is the smallest such constellation, i.e., if $\vect{\hat{\omega}} = \Omega(\vect{\hat{\omega}})$, then $\vect{\hat{\omega}} \succeq \vect{\omega}$.
\end{lem}
\begin{proof}
We have
\begin{align}
\hat{x} &= x_{\infty, \{\infty\}} = y_{1, \{\infty\}}^{(\infty)} \notag \\
&\stackrel{\eqref{eq_cwcfpfde}}{=} \epsilon g(y_{-\gamma + 2, \{\infty\}}^{(\infty)}, \cdots, y_{0, \{\infty\}}^{(\infty)}, y_{1, \{\infty\}}^{(\infty)}, \cdots, y_{\gamma, \{\infty\}}^{(\infty)}) \notag \\
&\stackrel{\eqref{eq_cwcfpfde}}{=} \epsilon g(\underbrace{x_{\infty, \{\infty\}}, \cdots, x_{\infty, \{\infty\}}}_{\gamma}, y_{2, \{\infty\}}^{(\infty)}, \cdots, y_{\gamma, \{\infty\}}^{(\infty)}) \notag \\
&= \epsilon g(\underbrace{\hat{x}, \cdots, \hat{x}}_{\gamma}, y_{2, \{\infty\}}^{(\infty)}, \cdots, y_{\gamma, \{\infty\}}^{(\infty)}). \notag
\end{align}
Hence, if we define $\vect{\omega}$ as follows
\[
\omega_i =
\begin{cases}
\hat{x} = y_{1, \{\infty\}}^{(\infty)}, & i \leq 1\\
y_{i, \{\infty\}}^{(\infty)}, & i > 1,
\end{cases}
\]
then it is clear that $\vect{\omega} = \pi(\vect{\omega}) = \Omega(\vect{\omega})$.

Note that any fixed point $\vect{\hat{\omega}}$ of the function $\Omega(\cdot)$ has to satisfy $\vect{\hat{\omega}} \succeq \vect{y}_{\{1\}}^{(\infty)}$ for the same channel erasure rate $\epsilon \in [0, 1]$ from the monotonicity of $g(\cdot)$. In particular, $\vect{\omega} \succeq \vect{y}_{\{1\}}^{(\infty)}$. From the continuity of the DE equations in Definition \ref{defn_wdfde}, it follows that $\vect{\omega}$ is the least solution to the equation $\vect{y} = \Omega(\vect{y})$, since it is the limiting constellation of the sequence of non-decreasing constellations $\{\vect{y}_{\{n\}}^{(\infty)}\}_{n = 1}^{\infty}$.\vspace*{3mm}
\end{proof}

We defer the proof of the following proposition to Appendix \ref{app_swthr}, which is the central argument in the proof of Theorem \ref{thm_wdthr}. Using the bound on the maximum transition width from Proposition \ref{prop_tw}, we obtain an upper bound on $\hat{x}$ for a given window size $W$ and erasure rate $\epsilon \in \mathscr{E}$. From this, we arrive at a lower bound for $\epsilon$ that guarantees $\hat{x} \leq \delta$ when $\delta$ is an arbitrarily chosen value smaller than $\delta_*$ (which depends only on $d_l, d_r$) and the window size is larger than $W_{\min}(\delta)$ (which depends only on the code parameters $d_l, d_r, \gamma$ and $\delta$). This gives us our lower bound on the WD threshold.
\begin{prop}[WD \& FW Thresholds] \label{prop_swthr} \vspace*{3mm}
Consider WD of the $(d_l, d_r, \gamma, L)$ spatially coupled ensemble with a window of size $W \ge W_{\min}(\delta) = \hat{W}_{\min}(\delta) + \gamma - 1$ over a BEC with erasure rate $\epsilon$. Then, we have
\begin{align}
\epsilon^\mrm{WD}(d_l, d_r, \gamma, W, \delta) \ge \Big(1 - &\frac{d_ld_r}{2}\delta^{\frac{d_l - 2}{d_l - 1}}\Big) \notag \\
&\times \epsilon^\mrm{FW}(d_l, d_r, \gamma, W - \gamma + 1, \delta) \notag
\end{align}
provided $\delta < \delta_* = \Big(\frac{2}{d_ld_r}\Big)^{\frac{d_l - 1}{d_l - 2}}$, where $\epsilon^\mrm{FW}(d_l, d_r, \gamma, W, \delta)$ is the first window threshold.\hfill$\blacksquare$\vspace*{3mm}
\end{prop}

From Proposition \ref{prop_swthr} and Equation \eqref{eq_fwthr}, we immediately have that
\begin{align}
\epsilon^\mrm{WD}(d_l, d_r, &\gamma, W, \delta) \ge \Big(1 - \frac{d_ld_r}{2}\delta^{\frac{d_l - 2}{d_l - 1}}\Big) \notag \\
&\times \Big(\epsilon^\mrm{BP}(d_l, d_r, \gamma) - e^{-\frac{1}{\mathsf{B}}(\frac{W - \gamma + 1}{\gamma - 1} - \mathsf{A}\ln\ln\frac{\mathsf{D}}{\delta} - \hat{\mathsf{C}})}\Big) \notag
\end{align}
provided $W \ge W_{\min}(\delta)$. By making the substitution $\mathsf{C} = \hat{\mathsf{C}} + 1$, we see that this proves Theorem \ref{thm_wdthr}.

\section{Experimental Results} \label{sec_expres}
In this section, we give results obtained by simulating windowed decoding of finite-length spatially coupled codes over the binary erasure channel. The code used for simulation was generated randomly by fixing the parameters $M = 1024$, $d_l = 3, d_r = 6$, with coupling length $\gamma = 3$ and chain length $L = 64$. The blocklength of the code was hence $n = ML = 65,536$ and the rate was $R \approx 0.484375$. From Table \ref{tab_bpthr}, the BP threshold for the ensemble to which this code belongs is $\epsilon^\mrm{BP}(d_l = 3, d_r = 6, \gamma = 3, L = 64) \approx 0.487514$.

Fig. \ref{fig_ber} shows the bit erasure rates achieved by using windows of length $W = 4, 6, 8$, i.e., the number of bits within each window was $WM = 4096, 6144$ and $8192$ respectively.
\begin{figure}[!ht]
\centering
\figselect{\include{berwd}}{\includegraphics{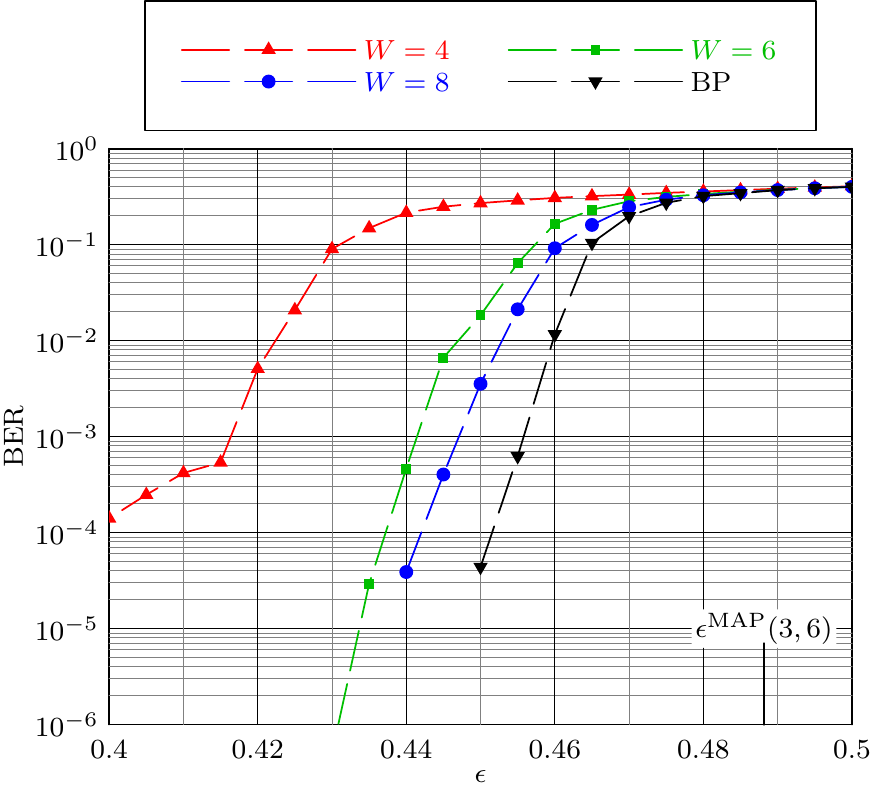}}
\caption{Bit erasure probability of the $(d_l = 3, d_r = 6, \gamma = 3, L = 64)$ spatially coupled code with $M = 1024$ achieved with a windowed decoder of window sizes $W = 4, 6$ and $8$.}
\label{fig_ber}
\end{figure}
From the figure, it is clear that good performance can be obtained for a wide range of channel erasure rates even for small window lengths, e.g., $W = 6, 8$. In performing the simulations above, we let the decoders (BP and WD) run for as many iterations as possible, until the decoder could solve for no further bits. For the windowed decoder, this meant that within each window configuration, the decoder was allowed to run until it could solve no further bits within the window. Fig. \ref{fig_iter} plots the average number of iterations for the BP decoder and the average number of iterations within each window configuration times the chain length (which corresponds to the average number of iterations) for the WD.
\begin{figure}[!ht]
\centering
\figselect{\include{iterwd}}{\includegraphics{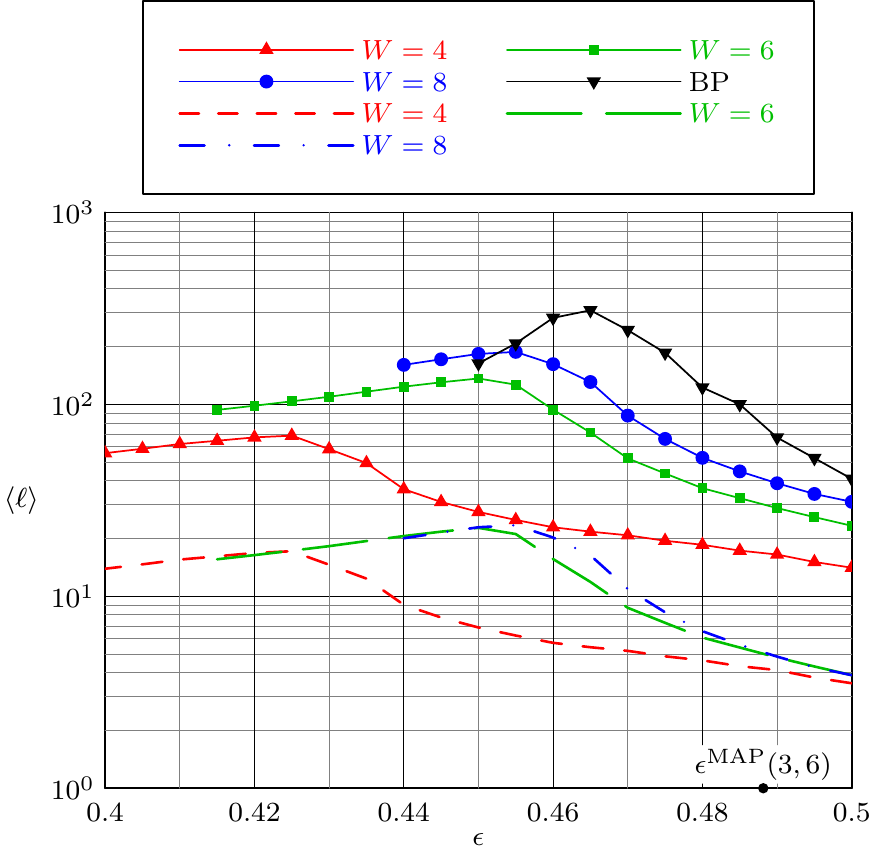}}
\caption{Average number of iterations $\langle\ell\rangle$ for BP and WD as a function of the channel erasure rate is shown for each window size in solid lines. For the WD, we show in dashed lines, the average number of iterations required within each window configuration.}
\label{fig_iter}
\end{figure}
We can see that for randomly chosen spatially coupled codes, a modest reduction in complexity is possible by using the windowed decoder in the waterfall region. Interestingly, the average number of iterations required per window configuration is independent of the chain length below certain channel erasure rates. 
The number of iterations required decreases beyond a certain value of $\epsilon$ because for these higher erasure rates, the decoder is no longer able to decode and gets stuck quickly. Although the smaller window sizes have a large reduction in complexity and a decent BER performance, the block erasure rate performance can be fairly bad, e.g., for the window of size $4$, the block erasure rate was $1$ in the range of erasure rates considered in Fig. \ref{fig_ber}. However, the block erasure rate improves drastically with increasing window size---for the window of size $8$, the block erasure rate at $\epsilon = 0.44$ was $\approx 6.3\times 10^{-4}$.

The above illustration suggests that 
for good performance with reduced complexity via windowed decoding, careful code design is necessary. For a certain variety of spatially coupled codes---protograph-based LDPC convolutional codes---certain design rules for good performance with windowed decoding were given in \cite{iye_12_tit_ldpccc}, and ensembles with good performance for a wide range of window sizes (including window sizes as small as $\gamma$) over erasure channels with and without memory were constructed. For these codes constructed using PEG \cite{hu_05_tit_peg} and ACE \cite{tia_04_toc_ace} techniques, not only can the error floor be lowered but also the performance of a medium-sized windowed decoder with fixed number of iterations can be made to be very close to that of the BP decoder \cite{iye_12_tit_ldpccc}. It is for such codes that the windowed decoder is able to attain very good performance with significant reduction in complexity and decoding latency.

\section{Conclusions} \label{sec_conc}
We considered a windowed decoding (WD) scheme for decoding spatially coupled codes that has smaller complexity and latency than the BP decoder. We analyzed the asymptotic performance limits of such a scheme by defining WD thresholds for meeting target erasure rates. We gave a lower bound on the WD thresholds and showed that these thresholds are guaranteed to approach the BP threshold for the spatially coupled code at least exponentially in the window size. Through density evolution, we showed that, in fact, the WD thresholds approach the BP threshold much faster than is guaranteed by the lower bound proved analytically. Since the BP thresholds for spatially coupled codes are themselves close to the MAP threshold, WD gives us an efficient way to trade off complexity and latency for decoding performance approaching the optimal MAP performance. Since the MAP decoder is capacity-achieving as the degrees of variables and checks are increased, similar performance is achievable through a WD scheme for a target erasure floor.

Through simulations, we showed that WD is a viable scheme for decoding finite-length spatially coupled codes and that even for small window sizes, good performance is attainable for a wide range of channel erasure rates. However, the complexity reduction for randomly constructed spatially coupled codes is not as significant as that obtained for protograph-based LDPC convolutional codes with a large girth. Thus, characterizing good spatially coupled codes within the ensemble of randomly coupled codes is a question that remains.

The WD scheme was analyzed here for the BEC and, therefore, the superior performance of these codes and the low complexity and latency of the WD scheme make these attractive for applications in coding over upper layers of the internet protocol. Furthermore, the same scheme can be employed for decoding spatially coupled codes over any channel. However, for channels that introduce errors apart from erasures, the WD scheme can suffer from error propagation. This effect would be similar to what occurs in decoding convolutional codes using a Viterbi decoder with a fixed traceback length. Analysis of the WD scheme and providing performance guarantees over such channels will play a key role in making spatially coupled codes and the WD scheme practical.

\appendices
\section{Proof of Lemma \ref{lem_boundsfp}} \label{app_boundsfp}
For the lower bound, we have
\begin{align}
y_i &= \epsilon g(y_{i - \gamma + 1}, \cdots, y_{i + \gamma - 1}) \notag \\
&\ge \epsilon g(\underbrace{0, \cdots, 0}_{\gamma}, \underbrace{y_{i + 1}, \cdots, y_{i + 1}}_{\gamma - 1}) \notag \\
&\stackrel{(a)}{\ge} \epsilon \Big(\frac{\gamma - 1}{2\gamma}y_{i + 1} \Big)^{d_l - 1} \notag
\end{align}
where $(a)$ follows from the fact \cite[Lem. 24(iii)]{kud_11_tit_ldpccc} that
\[
g(y_{i - \gamma + 1}, \cdots, y_{i + \gamma - 1}) \ge \bar{y}_i^{d_l - 1}
\]
where $\bar{y}_i = \frac{1}{\gamma^2}\sum_{j, k = 0}^{\gamma - 1}y_{i + j - k}$. Applying this bound recursively for $y_i, y_{i + 1}, \cdots, y_{i + j - 1}$, we get
\begin{align}
y_i &\ge \Big( \epsilon (\frac{\gamma - 1}{2\gamma})^{d_l - 1} \Big) ^ {\frac{(d_l - 1) ^ j - 1}{d_l - 2}}y_{i + j}^{(d_l - 1)^j} \triangleq \Phi e^{-\phi_j(d_l - 1)^j} \notag
\end{align}
where $\Phi = \Big( \epsilon(\frac{\gamma - 1}{2\gamma})^{d_l - 1} \Big)^{\frac{-1}{d_l - 2}} \ge 1$ and $\phi_j = \ln \Big( \frac{\Phi}{y_{i + j}} \Big) \ge 0$. When $i + j = W + 1$, since $y_{W + 1} = 1$,
\[
y_i \ge \Phi e^{-\phi(d_l - 1)^{W + 1 - i}},{\ }i = 1, 2, \cdots, W
\]
with $\phi = \phi_{W + 1 - i} = \ln \Phi$.

For the upper bound,
\begin{align}
y_i &\le \epsilon g(\underbrace{y_{i + k}, \cdots, y_{i + k}}_{\gamma + k}, \underbrace{1, \cdots, 1}_{\gamma - k - 1}) \notag \\
&\stackrel{(b)}{\le} \epsilon \Big(1 - \alpha_k(1 - y_{i + k})^{d_r - 1} \Big)^{d_l- 1} \notag \\
&\triangleq f_k(\epsilon, y_{i + k}) \notag
\end{align}
for $k \in [0, \gamma - 1]$, where $\alpha_k = (1 - \frac{(\gamma - k - 1)(\gamma - k)}{2\gamma^2})^{d_r - 1}$. Here $(b)$ follows from \cite[Lem. 24(i)]{kud_11_tit_ldpccc}
\[
g(y_{i - \gamma + 1}, \cdots, y_{i + \gamma - 1}) \le (1 - (1 - \bar{y}_i)^{d_r - 1})^{d_l - 1}.
\]
Note that for $k = \gamma - 1$, $f_{\gamma - 1}(\epsilon, x) = f(\epsilon, x)$, the forward DE update equation for the $(d_l, d_r)$-regular ensemble. This proves the Lemma. We now discuss the utility and limitations of the upper bounds derived here.

Fig. \ref{fig_ubiter} plots the bounds $f_k(\epsilon, y_{i + k})$ for the $(d_l = 3, d_r = 6, \gamma = 3, L)$ ensemble for two values of $\epsilon$, one below and the other above the BP threshold $\epsilon^{\mrm{BP}}(d_l, d_r)$.
\begin{figure}[!ht]
\centering
\figselect{\include{ubiterfig}}{\includegraphics{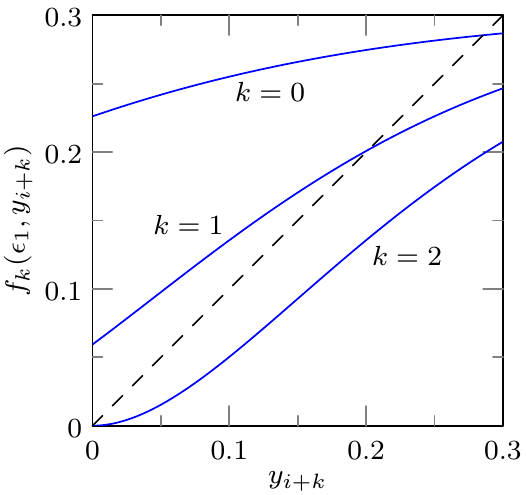}}
\figselect{\include{ubiterfig2}}{\includegraphics{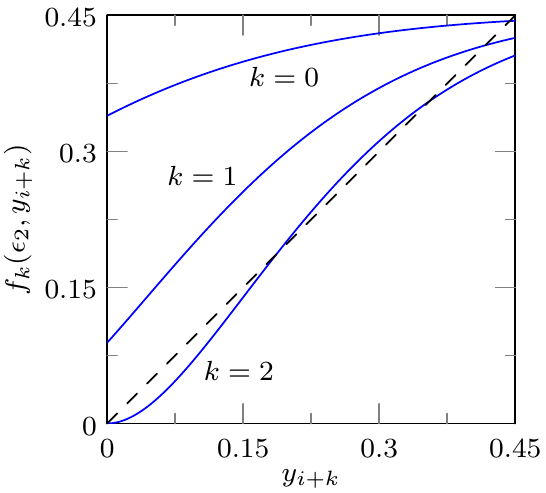}}
\caption{The upper bounds $f_k(\epsilon, y_{i + k})$ for two values of $\epsilon$ : $\epsilon_1 = 0.3 < \epsilon^{\mrm{BP}}(d_l, d_r) \approx 0.4294 < \epsilon_2 = 0.45$ for the $(d_l = 3, d_r = 6, \gamma = 3, L)$ ensemble.}
\label{fig_ubiter}
\end{figure}
As is clear from the figure, the tightest bounds are obtained for $k = \gamma - 1$. Note that the bound when $k = 0$ can be recursively computed to obtain a universal upper bound $y_\mrm{ub}$ on all the window constellation points $y_i$ for a given $(d_l, d_r, \gamma, L)$ ensemble, given by the fixed point of the equation
\[
y = f_0(\epsilon, y) = \epsilon \Big(1 - (\frac{\gamma + 1}{2\gamma})^{d_r - 1}(1 - y)^{d_r - 1} \Big)^{d_l- 1}
\]
which is plotted in Fig. \ref{fig_ubfp}. As can be seen from the plot, these upper bounds are only marginally tighter than the trivial upper bound of $\epsilon$.
\begin{figure}[!ht]
\centering
\figselect{\include{ubfp}}{\includegraphics{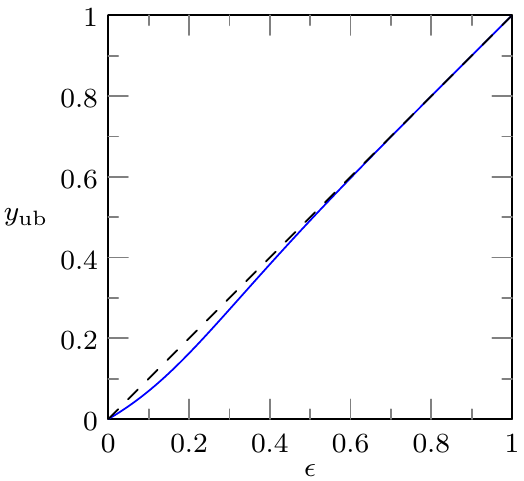}}
\caption{Universal upper bounds $y_\mrm{ub}$ on the constellation points $y_i$ as a function of $\epsilon$ for the $(d_l = 3, d_r = 6, \gamma = 2, L)$ ensemble. These bounds are only marginally tighter than the straightforward upper bound $\epsilon$. Also, the bounds are non-decreasing in $\gamma$.}
\label{fig_ubfp}
\end{figure}
In general, we can write $y_W \le y_\mrm{ub}$ and use the other upper bounds $f_k(\cdot, \cdot)$ to obtain better bounds for other sections as follows. In the sequel, we shall write $f_{k_1, k_2, \cdots, k_c}(\epsilon, y)$ to denote $f_{k_1}(\epsilon, f_{k_2}(\epsilon, \cdots f_{k_c}(\epsilon, y)))$ and similarly define
\[
f_{k^c}(\epsilon, y) \triangleq \underbrace{f_k(\epsilon, f_k(\epsilon, \cdots, f_k(\epsilon, y)))}_{c}.
\]
Thus, for $j = c(\gamma - 1) + d, 0 \le c, 0 \le d < \gamma - 1$, we can write $y_i \le f_{d,(\gamma - 1)^c}(\epsilon, y_{i + j})$. The FP value of the erasure probability of a variable node in the first section, $y_1$, can therefore be bounded in terms of the window size $W$ as
\[
y_1 \le f_{d, (\gamma - 1)^c}(\epsilon, y_\mrm{ub}),
\]
where $c = \left\lfloor\frac{W - 1}{\gamma - 1}\right\rfloor, d = W - 1 - c(\gamma - 1)$. This bounding is particularly useful when $\epsilon \le \epsilon^{\mrm{BP}}(d_l, d_r)$ when the fixed point of the $f_{\gamma - 1}(\cdot, \cdot)$ upper bound is zero. It is sometimes possible that $f_{(\gamma - 1)^c}(\epsilon, y_\mrm{ub}) \le f_{d, (\gamma - 1)^c}(\epsilon, y_\mrm{ub})$, in which case we can retain the tighter upper bound $f_{(\gamma - 1)^c}(\epsilon, y_\mrm{ub})$. Fig. \ref{fig_ubx1eg} shows an example of the upper bound on $y_1$ graphically.
\begin{figure}[!ht]
\centering
\figselect{\include{ubx1eg}}{\includegraphics{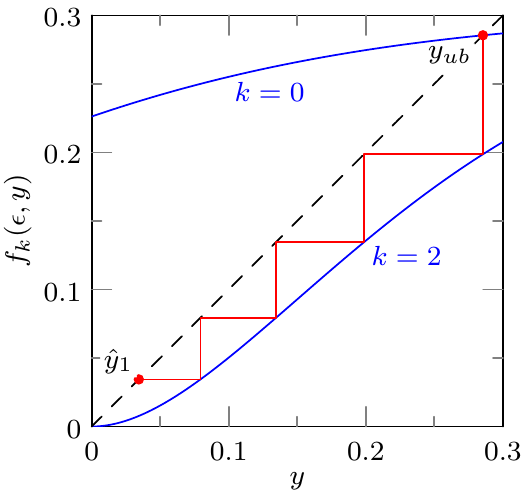}}
\caption{Upper bound, $\hat{y}_1 \ge y_1$, for the $(d_l = 3, d_r = 6, \gamma = 3, L)$ ensemble with a window of size $W = 9$. The channel erasure rate $\epsilon = 0.3$. Note that $c = 4, d = 0$, $y_{\text{ub}} \approx 0.285464$ and $\hat{y}_1 = f_{2^4}(0.3, y_{\text{ub}}) \approx 0.0343947$.}
\label{fig_ubx1eg}
\end{figure}
As a consequence of this upper bound, as $W \tends \infty$, we have that $y_1 \tends 0$ for $\epsilon \le \epsilon^{\mrm{BP}}(d_l, d_r)$. However, for $\epsilon > \epsilon^\mathrm{BP}(d_l, d_r)$, these upper bounds are not very useful since the FP of the $f_{\gamma - 1}(\cdot, \cdot)$ upper bound is non-zero (cf. Fig. \ref{fig_ubiter}).

\section{Proof of Proposition \ref{prop_tw}} \label{app_tw}
In the following, we will use some results from \cite{kud_11_tit_ldpccc} summarized below. We define
\[
h(y) \triangleq f(\epsilon, y) - y
\]
where
\[
f(\epsilon, y) = \epsilon(1 - (1 - y)^{d_r - 1})^{d_l - 1},
\]
the DE update equation for randomized $(d_l, d_r)$-regular ensembles. For $\epsilon \in (\epsilon^\mrm{BP}(d_l, d_r), 1)$, the equation $h(y) = 0$ has exactly three roots in the interval $[0, 1]$, given by $0, y_u(\epsilon)$ and $y_s(\epsilon)$. Between $0$ and $y_u(\epsilon)$, $h(y)$ is negative, attaining a unique minimum at $y_{{\min}}(\epsilon)$. Between $y_u(\epsilon)$ and $y_s(\epsilon)$, $h(y)$ is positive, attaining a unique maximum at $y_{{\max}}(\epsilon)$. Beyond $y_s(\epsilon)$, $h(y)$ is negative again. Between $0$ and $y_{{\min}}(\epsilon)$, $h(y)$ is upper bounded by a line through the origin with slope 
\[
-\mu_1(\epsilon) \triangleq \frac{h(y_{{\min}}(\epsilon))}{y_{{\min}}(\epsilon)},
\]
i.e., the line $l(y) = -\mu_1(\epsilon)y$. Between $y_{{\min}}(\epsilon)$ and $y_u(\epsilon)$, $h(y)$ is upper bounded by a line passing through $(y_u(\epsilon), 0)$ with a slope
\[
\mu_2(\epsilon) \triangleq \min\{\frac{-h(y_{{\min}}(\epsilon))}{y_u(\epsilon) - y_{{\min}}(\epsilon)}, h^\prime(y_u(\epsilon))\}. 
\]
Between $y_u(\epsilon)$ and $y_{{\max}}(\epsilon)$, $h(y)$ is lower bounded by a line through $(y_u(\epsilon), 0)$ with a slope
\[
\mu_3(\epsilon) \triangleq \min\{\frac{h(y_{{\max}}(\epsilon))}{y_{{\max}}(\epsilon) - y_u(\epsilon)}, h^\prime(y_u(\epsilon))\}. 
\]
Between $y_{{\max}}(\epsilon)$ and $y_s(\epsilon)$, $h(y)$ is lower bounded by the line through $(y_s(\epsilon), 0)$ with slope
\[
-\mu_4(\epsilon) \triangleq \max\{\frac{-h(y_{{\max}}(\epsilon))}{y_s(\epsilon) - y_{{\max}}(\epsilon)}, h^\prime(y_s(\epsilon))\}.
\]
Beyond $y_s(\epsilon)$, $h(y)$ is upper bounded by the line through $(y_s(\epsilon), 0)$ with slope
\[
-\mu_5(\epsilon) \triangleq h^\prime(y_s(\epsilon)).
\]
Each of the $\mu_i(\epsilon)$'s, $i = 1, \cdots, 5$, defined above is strictly positive for $\epsilon$ in the specified range. For a general $\epsilon$, we will drop the dependence of each of these parameters on $\epsilon$ from the notation. When $\epsilon = \epsilon^* \triangleq \epsilon^\mrm{MAP}(d_l, d_r)$, the corresponding parameters are themselves shown with $*$'s. These properties of $h(y)$ are illustrated in Fig. \ref{fig_hx}.
\begin{figure}[!ht]
\centering
\figselect{\include{hxfig}}{\includegraphics{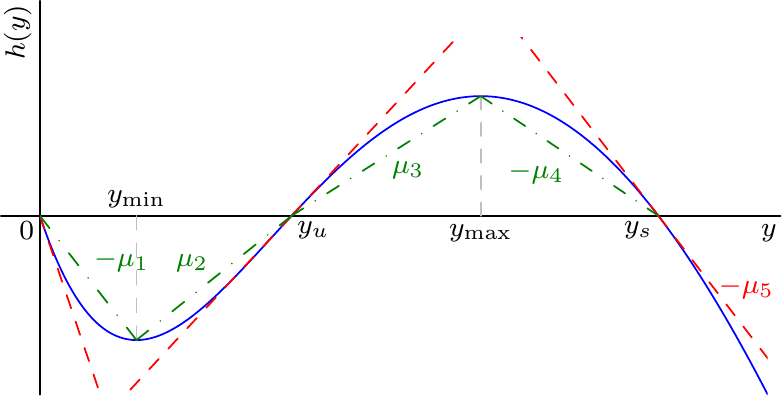}}
\caption{Plot of $h(y)$ (in solid blue) for the $(d_l = 3, d_r = 6)$ ensemble for $\epsilon = 0.47$ illustrating the properties stated above. We have dropped the dependence of all the parameters on $\epsilon$ from the notation. The tangents at $0, y_u$ and $y_s$ are shown as dashed red lines. The other lines used in bounding $h(y)$ are shown as dash-dotted green lines. The $\mu_i$'s, $i = 1, \cdots, 5$, are shown in the same color as the lines, whose absolute values of slopes they represent, that bound $h(y)$ in various regions.}
\label{fig_hx}
\end{figure}
We can lower bound $y_{\min}$ as $y_{\min} \geq \frac{1}{d_l^2d_r^2}$, and the slope $\mu_1$ as $\mu_1 \geq \frac{1}{8d_r^2} \triangleq \tilde{\mu}_1{\ }\forall{\ }\epsilon \in (\epsilon^\mrm{BP}(d_l, d_r), 1)$.

Further, $|h^\prime(y)| \le d_ld_r{\ }\forall{\ }y \in [0, 1]$. We have $h^\prime(0) = h^\prime(1) = -1$, $h^\prime(y_{\min}) = h^\prime(y_{\max}) = 0$. $h^{\prime\prime}(0) = h^{\prime\prime}(\hat{y}) = h^{\prime\prime}(1) = 0$, where
\[
\hat{y} = 1 - \Big( \frac{d_r - 2}{d_ld_r - d_l - d_r}\Big)^{\frac{1}{d_r - 1}},
\]
and
\[
h^{\prime\prime}(y) 
\begin{cases}
> 0, &y \in (0, \hat{y})\\
< 0, &y \in (\hat{y}, 1).
\end{cases}
\]
From Rolle's Theorem, $y_{\min} \le \hat{y} \le y_{\max}$. We first give some simple bounds for the $\mu_i$'s defined earlier which will be useful in the proof.

\begin{lem}[$\mu_4, \mu_5$ bounds] \label{lem_mubounds} \vspace*{3mm}
For $\epsilon \in (\epsilon^\mrm{BP}(d_l, d_r), 1)$, we have $0 < \mu_4 \le \mu_5 < 1$.
\end{lem}
\begin{proof}
Since $h^{\prime\prime}(y) < 0$ for $y \in (y_{\max}, 1)$, $h^\prime(y)$ monotonically decreases in this interval. Thus, $0 < \mu_5 = -h^\prime(y_s) < 1$. From the mean value theorem, we have $h^\prime(\xi) = -\frac{h(y_{\max})}{y_s - y_{\max}} \ge h^\prime(y_s)$ for some $\xi \in [y_{\max}, y_s]$ so that $0 < \mu_4 = \frac{h(y_{\max})}{y_s - y_{\max}} \le -h^\prime(y_s) = \mu_5 < 1$. \vspace*{3mm}
\end{proof}

The values $y_u$ and $y_s$ are referred to as the \emph{unstable} and \emph{stable} fixed points (FPs) of DE for the $(d_l, d_r)$-regular ensemble, respectively. This is because both these values satisfy $h(y) = 0$ or $y = f(\epsilon, y) = \epsilon(1 - (1 - y)^{d_r - 1})^{d_l - 1}$. The $\epsilon$ for which the FP is $y \in [0, 1]$ is given by
\[
\epsilon(y) = \frac{y}{(1 - (1 - y)^{d_r - 1})^{d_l - 1}}.
\]
The BP threshold is hence the smallest value of $\epsilon(y)$, i.e., $\epsilon^\mathrm{BP}(d_l, d_r) = \min \{\epsilon(y), y \in [0, 1]\}$. The value of $y$ that achieves this minimum is denoted $y^\mathrm{BP}$. Then, $\forall{\ }\epsilon \in [\epsilon^\mathrm{BP}(d_l, d_r), 1]$, the unstable and stable FPs are given by
\begin{align}
y_u(\epsilon) &= y \in [0, y^\mathrm{BP}] : \epsilon(y) = \epsilon, \notag \\
y_s(\epsilon) &= y \in [y^\mathrm{BP}, 1] : \epsilon(y) = \epsilon. \label{eq_yus}
\end{align}
Fig. \ref{fig_yus} plots these stable and unstable FPs.
\begin{figure}[!ht]
\centering
\figselect{\include{yusfig}}{\includegraphics{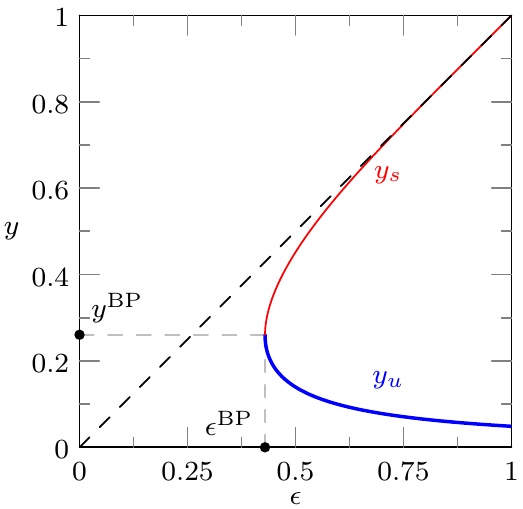}}
\caption{The unstable and stable FPs of DE for the $(d_l = 3, d_r = 6)$-regular ensemble as given by Equation \eqref{eq_yus}. $y_u(\epsilon)$ is shown as the thick blue curve and $y_s(\epsilon)$ as the thin red curve. By definition, $y_s(\epsilon) \geq y^\mathrm{BP}$ and $y_u(\epsilon) \leq y^\mathrm{BP}$. $y^\mathrm{BP}$ and $\epsilon^\mathrm{BP} \equiv \epsilon^\mathrm{BP}(d_l = 3, d_r = 6)$ are also shown. Note that $y_s(\epsilon^\mathrm{BP}) = y^\mathrm{BP} = y_u(\epsilon^\mathrm{BP})$.}
\label{fig_yus}
\end{figure}
The reason why $y_s$ is called the stable FP (and $y_u$ the unstable FP) can be explained through Fig. \ref{fig_yus}. For $\epsilon \in (\epsilon^\mathrm{BP}(d_l, d_r), 1]$, when the forward DE updates are performed, the value $y$ monotonically decreases from $1$ and converges to the first solution of the equation $h(y) = 0$, which happens to be $y_s(\epsilon)$ for $\epsilon$ in this range. Therefore, performing BP always results in the FP $y_s$ and hence the adjective ``stable''. Similarly for $y_u$, which is a solution never reached through BP, it can be shown that a small perturbation from the value of $y_u$ will result in convergence to either $y_s$ or $0$. Therefore, $y_u$'s are ``unstable'' FPs.

We can define the derivatives $y_s^\prime$ and $y_u^\prime$ of $y_s$ and $y_u$, respectively, with respect to $\epsilon$ for $\epsilon \in (\epsilon^\mathrm{BP}(d_l, d_r), 1)$. It is easy to see that $y_s^\prime$ is monotonically decreasing and $y_u^\prime$ is monotonically increasing in $\epsilon$. For details and proofs of the aforementioned properties, see \cite[Appendix II]{kud_11_tit_ldpccc}, \cite{ric_08_bok_mct}.

We are now ready to prove the proposition. Note that when $W$ is smaller than the claimed upper bound on the transition width, the claim is trivially true; i.e., the transition width cannot be longer than the window size. However, in this case, we cannot guarantee $y_1 \le \delta$. Hence, we will assume that $W$ is larger than the bound. In the following, we will often use the bound
\[
f(\epsilon, y_{i - \gamma + 1}) \leq y_i \leq f(\epsilon, y_{i + \gamma - 1}).
\]

We now define a schedule that results in a FP window constellation that dominates the FP of the parallel schedule, $\vect{y}$, for a channel erasure rate $\epsilon \in \msc{E}$. We then upper bound the actual transition width by the transition width of the dominating FP. We generate the dominating FP in steps.

\begin{enumerate}[i)]
\item \label{enum_first} Set $\vect{y}^{(0)} = \vect{1}$ and evaluate the sequence of window constellations $\{\vect{y}^{(\ell)}\}$ according to Equation \eqref{eq_wcfde}, but with the boundary conditions
\[
y_i^{(\ell)} =
\begin{cases}
y_1^{(\ell)}, &i \le 0\\
1, &i > W.
\end{cases}
\]
We have the FP in this case, $\vect{y}^A$, satisfying $\vect{y}^A \succeq \vect{y}$ by induction. Further,
\begin{align}
y_1^A &= \epsilon g(\underbrace{y_1^A, \cdots, y_1^A}_{\gamma}, y_2^A, \cdots, y_\gamma^A) \label{eq_proofi} \\
&\ge \epsilon g(\underbrace{y_1^A, \cdots, y_1^A}_{2\gamma - 1}) = f(\epsilon, y_1^A), \notag
\end{align}
so that $y_1^A \ge y_s$. Note that $y_1^A \leq y_u$ cannot happen since, starting from $1$, $y_1^A$ will equal the first solution of \eqref{eq_proofi}, which from the continuity of the DE equations is guaranteed to have a solution no smaller than $y_s$. Starting from the right end, we now count the number of sections until $y_i^A \le y_s + \Upsilon_A$ where we choose $\Upsilon_A = \frac{y_s^* - y_s}{2}$. Recall that the $*$-ed values correspond to $\epsilon^* = \epsilon^\mrm{MAP}(d_l, d_r)$. We first observe that
\[
y^A_{W - \gamma + 1} \le \epsilon g(y^A_W, \cdots, y^A_W) = f(\epsilon, y_W^A).
\]
Hence,
\begin{align}
y_W^A - y_{W - \gamma + 1}^A &\ge y_W^A - f(\epsilon, y_W^A) \notag \\
&= -h(y_W^A) \ge \mu_5(y_W^A - y_s), \notag
\end{align}
which implies that $y_{W - \gamma + 1}^A - y_s \le (1 - \mu_5)(y_W^A - y_s)$. From similar reasoning, we can show that
\begin{align}
y_{W - m(\gamma - 1)}^A - y_s &\le (1 - \mu_5)(y_{W - (m - 1)(\gamma - 1)}^A - y_s) \notag \\
&\le (1 - \mu_5)^m(y_W^A - y_s) \notag \\
&\le (1 - \mu_5)^m(y_\mrm{ub} - y_s). \notag
\end{align}
Since $\mu_5 < 1$ from Lemma \ref{lem_mubounds}, the above difference is decreasing in $m$. From the definition of $y_\mrm{ub}$ (note that this upper bound is valid even for the boundary conditions specified here) in the proof of Lemma \ref{lem_boundsfp}, it is easy to see that $y_\mrm{ub} \ge y_s$ so that the right hand side of the above chain of inequalities is non-negative. Thus, $y_{W - m(\gamma - 1)}^A \le y_s + \Upsilon_A$ if
\[
m \ge \left\lceil \frac{\ln\frac{y_\mrm{ub} - y_s}{\Upsilon_A}}{\ln\frac{1}{1 - \mu_5}} \right\rceil.
\]
Let $\tilde{\mu}_5 = \min \{\mu_5, \epsilon \in \msc{E}\}$. Then, we can write from the mean value theorem
\[
\Upsilon_A = \frac{y_s^* - y_s}{2} = \frac{\epsilon^* - \epsilon}{2}y_s^\prime(\hat{\epsilon})
\]
for some $\hat{\epsilon} \in [\epsilon, \epsilon^*]$. We can lower bound this as
\[
\Upsilon_A \ge \frac{\epsilon^* - \epsilon}{2}y_s^\prime(\epsilon^*) \ge \frac{\Delta\epsilon}{2}y_s^\prime(\epsilon^*)
\]
where the first inequality follows from the fact that $y_s^\prime$ is decreasing in $\epsilon$ in the interval $(\epsilon^\mrm{BP}(d_l, d_r), 1)$ and the second from $\epsilon^* \ge \epsilon^\mrm{BP}(d_l, d_r, \gamma)$. Therefore this width is no more than
\[
(\gamma - 1)\Big(\frac{\ln\frac{2(y_\mrm{ub} - y^\mrm{BP})}{y_s^\prime(\epsilon^*)\Delta\epsilon}}{\ln\frac{1}{1 - \tilde{\mu}_5}} + 1\Big) 
\]
sections, since $y^\mrm{BP} \le y_s$.

\item From the definition of $\Upsilon_A$, we have $y_s + \Upsilon_A = y_s^* - \Upsilon_A$. Let $i_A$ be the largest index for which $y_{i_A}^A \le y_s^* - \Upsilon_A$. Set $\vect{y}^{(0)} = \vect{y}^A$ and evaluate the sequence of window constellations $\{\vect{y}^{(\ell)}\}$ according to Equation \eqref{eq_wcfde} performing the updates only for those sections with indices $i < i_A$. Further, perform the updates for the channel erasure rate $\epsilon^*$ since we only require an upper bound on the transition width. We set the left end of the window to perform these updates to $0$, i.e., $y_i = 0{\ }\forall{\ }i \leq 0$. Let $\vect{y}^B$ denote the FP window constellation at the end of this procedure. By induction, we have $y_i^B \leq y_s^* - \Upsilon_A{\ }\forall{\ }i \le i_A$. Also, $y_{i_A - 1}^B \ge f(\epsilon^*, y_{i_A - \gamma}^B)$ so that
\begin{align}
y_{i_A - 1}^B - y_{i_A - \gamma}^B &\ge h(y_{i_A - \gamma}^B) \ge \mu_4^*(y_s^* - y_{i_A - \gamma}^B), \notag
\end{align}
where the last inequality assumes that $y_{i_A - \gamma}^B \geq y_{\max}^*$. Note that there is no loss of generality in this assumption, for if it were not true, we have that the number of sections with FP values $y_i^B$ between $y_{\max}^*$ and $y_s^*$ is smaller than the upper bound we derive in the following. The above inequality implies that 
\[
y_s^* - y_{i_A - \gamma}^B \ge \frac{y_s^* - y_{i_A - 1}^B}{1 - \mu_4^*}.
\]
Similarly, it can be shown that as long as $y_{i_A - 1 - m(\gamma - 1)}^B \ge y_{\max}^*$,
\[
y_s^* - y_{i_A - 1 - m(\gamma - 1)}^B \ge \frac{y_s^* - y_{i_A - 1 - (m - 1)(\gamma - 1)}^B}{1 - \mu_4^*}
\]
and by induction
\begin{align}
y_s^* - y_{i_A - 1 - m(\gamma - 1)}^B &\ge \frac{y_s^* - y_{i_A - 1}^B}{(1 - \mu_4^*)^m} \ge \frac{\Upsilon_A}{(1 - \mu_4^*)^m}. \notag
\end{align}
Note that since $\mu_4^* < 1$ from Lemma \ref{lem_mubounds}, the above difference is increasing in $m$. Thus, there are no more than
\[
(\gamma - 1)\Big(\frac{\ln\frac{2(y_s^* - y_{\max}^*)}{y_s^\prime(\epsilon^*)\Delta\epsilon}}{\ln\frac{1}{1 - \mu_4^*}} + 1\Big)
\]
sections with $y_i^B \in [y_{\max}^*, y_s^*]$.

\item Let $i_{\max}$ be the largest index $i$ such that $y_i^B \le y_{\max}^*$. We define $\Upsilon_B = \frac{y_u - y_u^*}{2}$ and count the number of sections with FP values $y_i^B$ between $y_u^* + \Upsilon_B$ and $y_{\max}^*$. Since $y_{i_{\max}}^B \ge f(\epsilon^*, y_{i_{\max} - (\gamma - 1)}^B)$ we have
\begin{align}
y_{i_{\max}}^B - y_{i_{\max} - (\gamma - 1)}^B &\ge h(y_{i_{\max} - (\gamma - 1)}^B) \notag \\
&\ge \mu_3^*(y_{i_{\max} - (\gamma - 1)}^B - y_u^*), \notag
\end{align}
where we again assume without loss of generality that $y_{i_{\max} - (\gamma - 1)}^B \geq y_u^*$. The above inequality implies that
\[
y_{i_{\max} - (\gamma - 1)}^B - y_u^* \le \frac{y_{i_{\max}}^B - y_u^*}{1 + \mu_3^*}.
\]
Again by induction,
\[
y_{i_{\max} - m(\gamma - 1)}^B - y_u^* \le \frac{y_{i_{\max}}^B - y_u^*}{(1 + \mu_3^*)^m}
\]
as long as $y^B_{i_{\max} - m(\gamma - 1)} \geq y_u^*$. Since $\mu_3^* > 0$, the above difference is decreasing in $m$, and consequently, $y_{i_{\max} - m(\gamma - 1)}^B \le y_u^* + \Upsilon_B$ if
\[
m \ge \left\lceil \frac{\ln\frac{y_{\max}^* - y_u^*}{\Upsilon_B}}{\ln(1 + \mu_3^*)} \right\rceil.
\]
Writing
\[
\Upsilon_B = \frac{y_u - y_u^*}{2} = -\frac{\epsilon^* - \epsilon}{2}y_u^\prime(\breve{\epsilon})
\]
from the mean value theorem for some $\breve{\epsilon} \in [\epsilon, \epsilon^*]$, we can bound this as
\[
\Upsilon_B \ge -\frac{\epsilon^* - \epsilon}{2}y_u^\prime(\epsilon^*) \ge -\frac{\Delta\epsilon}{2}y_u^\prime(\epsilon^*)
\]
where the first inequality follows because $-y_u^\prime$ is decreasing in $\epsilon$ in the interval $(\epsilon^\mrm{BP}(d_l, d_r), 1)$ and the second because $\epsilon^* \ge \epsilon^\mrm{BP}(d_l, d_r, \gamma)$. This implies that there are no more than
\[
(\gamma - 1)\Big(\frac{\ln\frac{2(y_{\max}^* - y_u^*)}{-y_u^\prime(\epsilon^*)\Delta\epsilon}}{\ln(1 + \mu_3^*)} + 1\Big)
\]
sections with FP values between $y_u^* + \Upsilon_B$ and $y_{\max}^*$.
\item From the definition of $\Upsilon_B$, we have $y_u^* + \Upsilon_B = y_u - \Upsilon_B$. Let $i_B$ be the largest index $i$ such that $y_i^B \le y_u - \Upsilon_B$. Set $\vect{y}^{(0)} = \vect{y}^B$ and evaluate $\{\vect{y}^{(\ell)}\}$ according to Equation \eqref{eq_wcfde} performing the updates only for sections with indices $i \le i_B$ with channel erasure rate $\epsilon$. Again we set the left end of the window to $0$ while performing the updates. Denote the FP obtained at the end of this procedure as $\vect{y}^C$. Clearly, $y_i^C \le y_u - \Upsilon_B{\ }\forall{\ }i \le i_B$. Since $y_{i_B - \gamma + 1}^C \le f(\epsilon, y_{i_B}^C)$, we have
\[
y_{i_B}^C - y_{i_B - \gamma + 1}^C \ge -h(y_{i_B}^C) \ge \mu_2(y_u - y_{i_B}^C)
\]
so that
\[
y_u - y_{i_B - \gamma + 1}^C \ge (1 + \mu_2)(y_u - y_{i_B}^C).
\]
Here, we assume that $y_{i_B}^C \geq y_{\min}$ in order to obtain an upper bound on the number of sections in the range $[y_{\min}, y_u]$. From similar reasoning as above, as long as $y_{i_B - m(\gamma - 1)}^C \ge y_{\min}$,
\[
y_u - y_{i_B - m(\gamma - 1)}^C \ge (1 + \mu_2)(y_u - y_{i_B - (m - 1)(\gamma - 1)}^C)
\]
and by induction
\begin{align}
y_u - y_{i_B - m(\gamma - 1)}^C &\ge (1 + \mu_2)^m(y_u - y_{i_B}^C) \notag \\
&\ge (1 + \mu_2)^m\Upsilon_B. \notag
\end{align}
Since $\mu_2 > 0$, the above difference is increasing in $m$. By letting $\tilde{\mu}_2 = \min\{\mu_2, \epsilon \in \msc{E}\}$ and noting that
\[
y_u - y_{\min} \le y^\mrm{BP} - y_{\min} \le y^\mrm{BP} - \frac{1}{d_l^2d_r^2},
\]
we have that there are no more than
\[
(\gamma - 1)\Big(\frac{\ln\frac{2(y^\mrm{BP} - \frac{1}{d_l^2d_r^2})}{-y_u^\prime(\epsilon^*)\Delta\epsilon}}{\ln(1 + \tilde{\mu}_2)} + 1\Big)
\]
sections with FP values in the interval $[y_{\min}, y_u]$.
\item Let $i_C$ be the largest index $i$ such that $y_i^C \le y_{\min}$. Proceeding as above, we have
\begin{align}
y_{i_C}^C - y_{i_C - \gamma + 1}^C &\ge -h(y_{i_C}^C) \ge \mu_1y_{i_C}^C \notag
\end{align}
and by induction,
\[
y_{i_C - m(\gamma - 1)}^C \le (1 - \mu_1)^my_{i_C}^C.
\]
Thus, between $\delta_0$ and $y_{\min}$, there are no more than
\[
(\gamma - 1)\Big(\frac{\ln\frac{y^\mrm{BP}}{\delta_0}}{\ln\frac{1}{1 - \tilde{\mu}_1}} + 1\Big)
\]
sections with FP values in the interval $[\delta_0, y_{\min}]$, since $y_{\min} \le y^\mrm{BP}$ and $\mu_1 \ge \tilde{\mu}_1$.

\begin{figure*}[t]
\begin{align}
\hat{\mathsf{C}} &= \frac{\ln\frac{1}{\tilde{\psi}}}{\ln(d_l - 1)} + \frac{\ln\frac{y^\mrm{BP}}{\delta_0}}{\ln\frac{1}{1 - \tilde{\mu}_1}} + \frac{\ln\frac{2(y^\mrm{BP} - \frac{1}{d_l^2d_r^2})}{-y_u^\prime(\epsilon^*)}}{\ln(1 + \tilde{\mu}_2)} + \frac{\ln\frac{2(y_{\max}^* - y_u^*)}{-y_u^\prime(\epsilon^*)}}{\ln(1 + \mu_3^*)} + \frac{\ln\frac{2(y_s^* - y_{\max}^*)}{y_s^\prime(\epsilon^*)}}{\ln\frac{1}{1 - \mu_4^*}} + \frac{\ln\frac{2(y_\mrm{ub} - y^\mrm{BP})}{y_s^\prime(\epsilon^*)}}{\ln\frac{1}{1 - \tilde{\mu}_5}} + 6. \label{eq_constc}
\end{align}
\hrule
\end{figure*}

\item \label{enum_last} Let $i_D$ be the largest index $i$ such that $y_i^C \le \delta_0$. From Lemma \ref{lem_detail}, we know that the tail decays doubly-exponentially for $i \le i_D$. From \eqref{eq_detail}, we have $y_{i_D - m(\gamma - 1)}^C \le \Psi e^{-\psi(d_l - 1)^m} \le \tilde{\Psi}e^{-\tilde{\psi}(d_l - 1)^m}$, where
\begin{align}
\Psi &= \delta_0\epsilon^{-\frac{1}{d_l - 2}} \notag \\
&\le \delta_0(\frac{\epsilon^\mrm{BP}(d_l, d_r, \gamma) + \epsilon^\mrm{BP}(d_l, d_r)}{2})^{-\frac{1}{d_l - 2}} \triangleq \tilde{\Psi} \notag
\end{align}
and
\[
\psi = \frac{1}{d_l - 2}\ln\frac{1}{\epsilon} \ge \frac{1}{d_l - 2}\ln\frac{1}{\epsilon^\mrm{MAP}(d_l, d_r)} \triangleq \tilde{\psi}.
\]
Thus there are no more than
\[
(\gamma - 1)\Big(\frac{1}{\ln(d_l - 1)}\ln\ln\frac{\tilde{\Psi}}{\delta} + \frac{\ln\frac{1}{\tilde{\psi}}}{\ln(d_l - 1)} + 1\Big)
\]
sections with $y_i^C \in [\delta, \delta_0]$.
\end{enumerate}

Finally, collecting all these terms, we conclude that the transition width of the FP obtained from the procedure highlighted in the steps \ref{enum_first}) through \ref{enum_last}) is upper bounded by
\[
\hat{\tau}(\epsilon, \delta) = (\gamma - 1)\Big(\mathsf{A}\ln\ln\frac{\mathsf{D}}{\delta} + \mathsf{B}\ln\frac{1}{\Delta\epsilon} + \hat{\mathsf{C}}\Big)
\]
where the constants $\mathsf{A}, \mathsf{B}$ and $\mathsf{D}$ are as follows
\begin{align}
\mathsf{A} &= \frac{1}{\ln(d_l - 1)}, \notag \\
\mathsf{B} = \frac{1}{\ln(1 + \tilde{\mu}_2)} + &\frac{1}{\ln(1 + \mu_3^*)} + \frac{1}{\ln(\frac{1}{1 - \mu_4^*})} + \frac{1}{\ln(\frac{1}{1 - \tilde{\mu}_5})}, \notag \\
&\hspace*{5mm}\mathsf{D} = \tilde{\Psi}. \notag
\end{align}
The constant $\hat{\mathsf{C}}$ is as given in Equation \eqref{eq_constc}. Note that these constants depend only on the ensemble parameters $d_l, d_r$ and $\gamma$. Since it is clear that the FP obtained through the procedure in steps \ref{enum_first}) through \ref{enum_last}) above dominates pointwise the first window configuration FP of forward DE with a window of size $W$ for channel erasure rate $\epsilon$, we can guarantee that the transition width is upper bounded by the above expression. This completes the proof.

\section{Proof of Proposition \ref{prop_swthr}} \label{app_swthr}
We start with the first window configuration FP of forward DE when the channel erasure rate is $\epsilon$ and show that this FP dominates the $c^\text{th}$ window configuration FP of forward DE for every $c$ for a smaller channel erasure rate $\upsilon \leq \epsilon$. To prove this, it suffices to show that the FP $\vect{\omega}$ defined in Lemma \ref{lem_hatx} for channel erasure rate $\upsilon$ is dominated pointwise by the first window configuration FP for channel erasure rate $\epsilon$. This establishes $\upsilon$ as being a lower bound on the WD threshold $\epsilon^\mrm{WD}$.

Set $\vect{y}^{(0)} = \vect{y}_{\{1\}}^{(\infty)}$, the first window configuration FP of forward DE for channel erasure rate $\epsilon$. Evaluate $\{\vect{y}^{(\ell)}\}_{\ell = 1}^{\infty}$ according to
\[
\vect{y}^{(\ell)} = \Omega(\vect{y}^{(\ell - 1)})
\]
where $\Omega(\cdot)$ is as defined in Lemma \ref{lem_hatx}, but for channel erasure rate $\upsilon = \epsilon - \Delta\epsilon$. Then, the following are true:
\[
y_i^{(0)} = 
\begin{cases}
\epsilon g(\underbrace{0, \cdots, 0}_{\gamma - i}, y_1^{(0)}, \cdots, y_{i + \gamma - 1}^{(0)}), &1 \le i < \gamma \\
\epsilon g(y_{i - \gamma + 1}^{(0)}, \cdots, y_{i + \gamma - 1}^{(0)}), &\gamma \le i \le W
\end{cases}
\]
and
\[
y_i^{(1)} = 
\begin{cases}
\upsilon g(\underbrace{y_1^{(0)}, \cdots, y_1^{(0)}}_{\gamma - i}, y_1^{(0)}, \cdots, y_{i + \gamma - 1}^{(0)}), &1 \le i < \gamma\\
\upsilon g(y_{i - \gamma + 1}^{(0)}, \cdots, y_{i + \gamma - 1}^{(0)}), &\gamma \le i \le W.
\end{cases}
\]
For $\gamma \le i \le W$,
\begin{equation} \label{eq_fpineq}
y_i^{(1)} = \frac{\upsilon}{\epsilon}y_i^{(0)} = \frac{\epsilon - \Delta\epsilon}{\epsilon}y_i^{(0)} \leq y_i^{(0)}. 
\end{equation}
Let us write
\[
g_i(\sigma, y_1, \cdots, y_{i + \gamma - 1}) \triangleq g(\underbrace{\sigma, \cdots, \sigma}_{\gamma - i}, y_1, \cdots, y_{i + \gamma - 1})
\]
and
\[
G_i(\epsilon, \sigma, y_1, \cdots, y_{i + \gamma - 1}) \triangleq \epsilon g_i(\sigma, y_1, \cdots, y_{i + \gamma - 1})
\]
for $1 \le i < \gamma$. For $i$ in this range, consider
\begin{align}
y_i^{(0)} - y_i^{(1)} &= G_i(\epsilon, 0, y_1^{(0)}, \cdots, y_{i + \gamma - 1}^{(0)}) \notag \\
&\hspace*{5mm} - G_i(\upsilon, y_1^{(0)}, y_1^{(0)}, \cdots, y_{i + \gamma - 1}^{(0)}) \notag \\
&= \Big[ G_i(\epsilon, 0, y_1^{(0)}, \cdots, y_{i + \gamma - 1}^{(0)}) \notag \\
&\hspace*{7mm} - G_i(\upsilon, 0, y_1^{(0)}, \cdots, y_{i + \gamma - 1}^{(0)}) \Big] \notag \\
&\hspace*{4mm} - \Big[ G_i(\upsilon, y_1^{(0)}, y_1^{(0)}, \cdots, y_{i + \gamma - 1}^{(0)}) \notag \\
&\hspace*{11mm} - G_i(\upsilon, 0, y_1^{(0)}, \cdots, y_{i + \gamma - 1}^{(0)}) \Big] \notag \\
&\stackrel{(a)}{=} \Delta\epsilon\frac{\partial G_i}{\partial\epsilon}\Big|_{\xi, \sigma = 0, \vect{y}^{(0)}} - y_1^{(0)}\frac{\partial G_i}{\partial\sigma}\Big|_{\upsilon, \zeta, \vect{y}^{(0)}} \label{eq_lbdiff}
\end{align}
where $\xi \in [\upsilon, \epsilon]$ and $\zeta \in [0, y_1^{(0)}]$. Here, $(a)$ follows from the mean value theorem. We have
\[
\frac{\partial G_i}{\partial\epsilon}\Big|_{\xi, \sigma = 0, \vect{y}^{(0)}} = g_i(\sigma, y_1, \cdots, y_{i + \gamma - 1})\Big|_{\xi, \sigma = 0, \vect{y}^{(0)}} = \frac{y_i^{(0)}}{\epsilon}.
\]
Since $\frac{\partial G_i}{\partial \sigma} = \epsilon \frac{\partial g_i}{\partial \sigma}$, we focus on $g_i$. Expanding out the expression for $g_i$, it can be written as
\begin{align}
g_i &= \Big[1 - \frac{1}{\gamma}\Big((\alpha_{i, 1} - \frac{\gamma - i}{\gamma}\sigma)^{d_r - 1} + \cdots  \notag \\
&\hspace*{4mm} + (\alpha_{i, \gamma - i} - \frac{1}{\gamma}\sigma)^{d_r - 1} + \alpha_{i, \gamma - i + 1}^{d_r - 1} + \cdots + \alpha_{i, \gamma}^{d_r - 1}\Big)\Big]^{d_l - 1} \notag
\end{align}
where
\[
\alpha_{i, j + 1} =
\begin{cases}
1 - \frac{1}{\gamma}\sum_{c = 1}^{i + j}y_c, & 0 \leq j \leq \gamma - i - 1\\ 
1 - \frac{1}{\gamma}\sum_{c = i + j - \gamma + 1}^{i + j}y_c, & \gamma - i \leq j \leq \gamma - 1. 
\end{cases}
\]
Clearly, $0 \leq \alpha_{i, j + 1} \leq 1{\ }\forall{\ }1 \leq i < \gamma, 0 \leq j \leq \gamma - 1$. Therefore,
\begin{align}
\frac{\partial g_i}{\partial\sigma} &= \frac{(d_l - 1)(d_r - 1)}{\gamma}g_i^{\frac{d_l - 2}{d_l - 1}}\Big[(\alpha_{i, 1} - \frac{\gamma - i}{\gamma}\sigma)^{d_r - 2}\frac{\gamma - i}{\gamma} \notag \\
&\hspace*{10mm} + \cdots + (\alpha_{i, \gamma - i} - \frac{1}{\gamma}\sigma)^{d_r - 2}\frac{1}{\gamma}\Big] \notag \\
&\stackrel{(a)}{\leq} \frac{d_ld_r}{\gamma^2}g_i^{\frac{d_l - 2}{d_l - 1}}\frac{(\gamma - i)(\gamma - i + 1)}{2} \le \frac{d_ld_r}{2}g_i^{\frac{d_l - 2}{d_l - 1}}. \notag
\end{align}
Here, $(a)$ holds because $0 \leq (\alpha_{i, j + 1} - \frac{\gamma - i - j}{\gamma}\sigma) \leq 1{\ }\forall{\ }0\leq j \leq \gamma - i - 1$. This implies that
\begin{align}
\frac{\partial G_i}{\partial\sigma}\Big|_{\upsilon, \zeta, \vect{y}^{(0)}} &\leq \frac{d_ld_r}{2}g_i^{\frac{d_l - 2}{d_l - 1}}\epsilon\Big|_{\upsilon, \zeta, \vect{y}^{(0)}} \notag \\
&= \frac{d_ld_r}{2}\upsilon^{\frac{1}{d_l - 1}} \notag \\
&\hspace*{7mm}\times \Big(\upsilon g(\underbrace{\zeta, \cdots, \zeta}_{\gamma - i}, y_1^{(0)}, \cdots, y_{i + \gamma - 1}^{(0)})\Big)^{\frac{d_l - 2}{d_l - 1}} \notag \\
&\stackrel{(b)}{\leq} \frac{d_ld_r}{2}\upsilon^{\frac{1}{d_l - 1}} \notag \\
&\hspace*{2mm}\times \Big(\upsilon g(\underbrace{y_1^{(0)}, \cdots, y_1^{(0)}}_{\gamma - i}, y_1^{(0)}, \cdots, y_{i + \gamma - 1}^{(0)})\Big)^{\frac{d_l - 2}{d_l - 1}} \notag \\
&= \frac{d_ld_r}{2}\upsilon^{\frac{1}{d_l - 1}}(y_i^{(1)})^{\frac{d_l - 2}{d_l - 1}} \notag \\
&\stackrel{(c)}{\le} \frac{d_ld_r}{2}\upsilon^{\frac{1}{d_l - 1}}(y_\gamma^{(1)})^{\frac{d_l - 2}{d_l - 1}} \le \frac{d_ld_r}{2}(y_\gamma^{(1)})^{\frac{d_l - 2}{d_l - 1}} \notag \\
&\stackrel{\eqref{eq_fpineq}}{\le} \frac{d_ld_r}{2}y_\gamma^{(0)^{\frac{d_l - 2}{d_l - 1}}}. \notag
\end{align}
Here, the inequality labeled $(b)$ is true because $\zeta \leq y_1^{(0)}$, $(c)$ follows from the observation that $y^{(1)}_i \leq y_{i + 1}^{(1)}, i \geq 1$, which is in turn true since $\vect{y}^{(0)}$ and $\pi(\vect{y}^{(0)})$ were non-decreasing. Substituting back in \eqref{eq_lbdiff}, we have for $1 \leq i < \gamma$
\[
y_i^{(0)} - y_i^{(1)} \ge \Delta\epsilon\frac{y_1^{(0)}}{\epsilon} - y_1^{(0)}\frac{d_ld_r}{2}y_\gamma^{(0)^{\frac{d_l - 2}{d_l - 1}}}.
\]
Thus if $\frac{\Delta\epsilon}{\epsilon} \ge \frac{d_ld_r}{2}y_\gamma^{(0)^{\frac{d_l - 2}{d_l - 1}}}$, $y_i^{(0)} \geq y_i^{(1)}{\ }\forall{\ }i \geq 1$, and hence $\pi(\vect{y}^{(0)})$ dominates $\vect{y}^{(1)}$ pointwise. Recall that
\[
\pi(y_i) =
\begin{cases}
y_1, & i < 1\\
y_i, & i \geq 1.
\end{cases}
\]
It therefore follows by induction that the limiting constellation $\vect{y}^{(\infty)}$ exists, and is also dominated by $\pi(\vect{y}^{(0)})$. It is clear that $\vect{y}^{(\infty)}$ satisfies
\[
\vect{y}^{(\infty)} = \Omega(\vect{y}^{(\infty)}).
\]
From Lemma \ref{lem_hatx}, $\vect{y}^{(\infty)} \succeq \vect{\omega}$ and hence $y_1^{(\infty)} \geq \hat{x}$.

If the window size is chosen to be $W \ge \hat{W}_{\min}(\delta) + \gamma - 1 \triangleq W_{\min}(\delta)$, then for the first window, we can guarantee $y_\gamma^{(0)} \le \delta$ for some $\delta < \delta_0$ for all channel erasure rates smaller than $\epsilon^\mrm{FW} \equiv \epsilon^\mrm{FW}(d_l, d_r, \gamma, W - \gamma + 1, \delta)$. From the above argument, it follows that we can ensure $\hat{x} \leq \delta$ for all erasure rates smaller than $\epsilon^\mrm{FW}\Big(1 - \frac{d_ld_r}{2}\delta^{\frac{d_l - 2}{d_l - 1}}\Big)$. As long as
\[
\delta <  \delta_* \triangleq \Big(\frac{2}{d_ld_r}\Big)^{\frac{d_l - 1}{d_l - 2}} < \Big(\frac{1}{d_r - 1}\Big)^{\frac{d_l - 1}{d_l - 2}} = \delta_0,
\]
this erasure rate is a non-trivial lower bound on the WD threshold $\epsilon^\mrm{WD}$.

\section*{Acknowledgment}
The authors thank the anonymous reviewer for help to improve the presentation of the paper. A. R. Iyengar would like to thank S. Kudekar for pointing out Lemma \ref{lem_detail} and for helpful suggestions in proving parts of Proposition \ref{prop_tw}.

\twobibs{
\bibliographystyle{IEEEtran}
\bibliography{/media/sg/work/ucsd/Bib/mybib}
}
{

}

\begin{IEEEbiographynophoto}{Aravind R. Iyengar} (S'09-M'12) 
 received his B.Tech degree in Electrical Engineering from the Indian Institute of Technology Madras, Chennai, in 2007; his M.S. and Ph.D. degrees in Electrical Engineering from the University of California in San Diego, La Jolla, where he was affiliated with the Center for Magnetic Recording Research, in 2009 and 2012 respectively. He is currently with Qualcomm Technologies Inc., Santa Clara, where he is involved in the design of baseband modems. In 2006, he was a visiting student intern at the \'{E}cole Nationale Sup\'{e}rieure de l'Electronique et de ses Applications (ENSEA), Cergy, France. He was a visiting doctoral student at the Communication Theory Laboratory at the \'{E}cole Polytechnique F\'{e}d\'{e}rale de Lausanne (EPFL), Lausanne, Switzerland in 2010. His research interests are in the areas of information and coding theory, and in signal processing and wireless communications.

A. R. Iyengar was the recipient of the Sheldon Schultz Prize for Excellence in Graduate Student Research at the University of California, San Diego in 2012.
\end{IEEEbiographynophoto}

\begin{IEEEbiographynophoto}{Paul H. Siegel} (M'82-SM'90-F'97) received the S.B. and Ph.D. degrees in mathematics from  the Massachusetts Institute of Technology (MIT), Cambridge, in 1975 and 1979, respectively.

He held a Chaim Weizmann Postdoctoral Fellowship at the Courant Institute, New York University. He was with the IBM Research Division in San Jose, CA, from 1980 to 1995. He joined the faculty at the University of California, San Diego in July 1995, where he is currently Professor o
f Electrical and Computer Engineering in the Jacobs School of Engineering. He is affiliated with the Center for Magnetic Recording Research where he holds an endowed chair and served as Director from 2000 to 2011. His primary research interests lie in the areas of information theory and communications, particularly coding and modulation techniques, with applications to digital data storage and transmission.

Prof. Siegel was a member of the Board of Governors of the IEEE Information Theory Society from 1991 to 1996 and from 2009 to 2011. He was re-elected for another 3-year term in 2012. He served as Co-Guest Editor of the May 1991 Special Issue on ``Coding for Storage Devices'' of the IEEE Transactions on Information Theory. He served the same Transactions as Associate Editor for Coding Techniques from 1992 to 1995, and as Editor-in-Chief from July 2001 to July 2004. He was also Co-Guest Editor of the May/September 2001 two-part issue on ``The Turbo Principle: From Theory to Practice'' of the IEEE Journal on Selected Areas in Communications.

Prof. Siegel was co-recipient, with R. Karabed, of the 1992 IEEE Information Theory Society Paper Award and shared the 1993 IEEE Communications Society Leonard G. Abraham Prize Paper Award with B.H. Marcus and J.K. Wolf. With J.B. Soriaga and H.D. Pfister, he received the 2007 Best Paper Award in Signal Processing and Coding for Data Storage from the Data Storage Technical Committee of the IEEE Communications Society. He holds several patents in the area of coding and detection, and was named a Master Inventor at IBM Research in 1994.  He is an IEEE Fellow and a member of the National Academy of Engineering.
\end{IEEEbiographynophoto}

\begin{IEEEbiographynophoto}{R\"{u}diger L. Urbanke} received the Diplomingenieur degree from the Vienna Institute of Technology, Vienna, Austria, in 1990 and the M.Sc. and PhD degrees in electrical engineering from Washington University, St. Louis, MO, in 1992 and 1995 respectively.

From 1995 to 1999, he held a position at the Mathematics of Communications Department at Bell Labs. Since November 1999, he has been a faculty member at the School of Computer \& Communication Sciences of EPFL, Lausanne, Switzerland, where he is the head of the Communications Theory Lab as well as the head of the Doctoral Program of the School of Computer and Communication Sciences (comprising roughly 250 PhD students).

Dr. Urbanke's research is focused on the analysis and design of coding systems and, more generally, graphical models.

Dr. Urbanke is a recipient of a Fulbright Scholarship. From 2000-2004 he was an Associate Editor of the IEEE Transactions on Information Theory and he has been elected in October 2012 to the Board of Governors of IEEE Information Theory Society. He is also currently on the board of the series ``Foundations and Trends in Communications and Information Theory.'' He is a co-recipient of the IEEE Information Theory Society 2002 Best Paper Award and a co-recipient of the 2011 IEEE Kobayashi Computers and Communications Award. He is co-author of the book ``Modern Coding Theory'' published by Cambridge University Press.
\end{IEEEbiographynophoto}

\begin{IEEEbiographynophoto}{Jack Keil Wolf} (S'54-M'60-F'73-LF'97) received the B.S.E.E. degree from the University of Pennsylvania Philadelphia, in 1956, and the M.S.E., M.A., and Ph.D. degrees from Princeton University, Princeton, NJ, in 1957, 1958, and 1960, respectively. He was the Stephen O. Rice Professor of Electrical and Computer Engineering and a member of the Center for Magnetic Recording Research at the University of California-San Diego, La Jolla. He was a member of the Electrical Engineering Department at New York University from 1963 to 1965, and the Polytechnic Institute of Brooklyn from 1965 to 1973. He was Chairman of the Department of Electrical and Computer Engineering at the University of Massachusetts, Boston, from 1973 to 1975, and he was Professor there from 1973 to 1984. From 1984 to 2011, he was a Professor of Electrical and Computer Engineering and a member of the Center for Magnetic Recording Research at the University of California-San Diego. He also held a part-time appointment at Qualcomm, Inc., San Diego. From 1971 to 1972, he was an NSF Senior Postdoctoral Fellow, and from 1979 to 1980, he held a Guggenheim Fellowship. His most recent research interest was in signal processing for storage systems.

Dr. Wolf was elected to the National Academy of Engineering in 1993. He was the recipient of the 1990 E. H. Armstrong Achievement Award of the IEEE Communications Society and was co-recipient with D. Slepian of the 1975 IEEE Information Theory Group Paper Award for the paper ``Noiseless coding for correlated information sources.'' He shared the 1993 IEEE Communications Society Leonard G. Abraham Prize Paper Award with B. Marcus and P.H. Siegel for the paper ``Finite-State Modulation Codes for Data Storage.'' He served on the Board of Governors of the IEEE Information Theory Group from 1970 to 1976 and from 1980 to 1986. Dr. Wolf was President of the IEEE Information Theory Group in 1974. He was International Chairman of Committee C of URSI from 1980 to 1983. He was the recipient of the 1998 IEEE Koji Kobayashi Computers and Communications Award, ``for fundamental contributions to multi-user communications and applications of coding theory to magnetic data storage devices.'' In May 2000, he received a UCSD Distinguished Teaching Award. In 2004 Professor Wolf received the IEEE Richard W. Hamming Medal for ``fundamental contributions to the theory and practice of information transmission and storage.'' In 2005 he was elected by the American Academy of Arts and Sciences as a Fellow, and in 2010 was elected as a member of the National Academy of Sciences. He was co-recipient with I.M. Jacobs of the 2011 Marconi Society Fellowship and Prize. Prof. Wolf passed away on May 12, 2011.
\end{IEEEbiographynophoto}

\end{document}